\RequirePackage{luatex85}

\documentclass[11pt,a4paper]{article}
\usepackage{iftex}

\ifPDFTeX
\usepackage[utf8]{inputenc}
\usepackage[T1]{fontenc}
\fi

\usepackage{amsmath}
\usepackage{amssymb}
\usepackage{amsthm}
\usepackage{thmtools}

\ifPDFTeX
\usepackage{libertine}
\usepackage[libertine]{newtxmath} % must load after ams packages
\usepackage{MnSymbol}
\usepackage[scaled=0.83]{beramono} % use Bera Mono as the typewriter font
\else
\usepackage{unicode-math}
\fi

\usepackage[DIV=10]{typearea} % large value = less margin, recommended values: 10pt: 8, 11pt: 10, 12pt: 12

\usepackage{microtype}

\usepackage[procnumbered,ruled,vlined,linesnumbered]{algorithm2e}

\usepackage{xcolor}
\usepackage{xspace}
\usepackage{enumitem}

\usepackage[
	style=alphabetic,
	backref=true,
	doi=false,
	url=false,
	maxcitenames=3,
	mincitenames=3,
	maxbibnames=10,
	minbibnames=10,
	backend=biber,
	sortlocale=en_US
]{biblatex}

\usepackage[ocgcolorlinks]{hyperref} % Option ocgcolorlinks makes links only colored on screen and not on printouts
\usepackage{cleveref}

\newcommand*\samethanks[1][\value{footnote}]{\footnotemark[#1]}

\allowdisplaybreaks[1]

\renewcommand{\epsilon}{\varepsilon}

\protected\def\mathbb#1{\text{\usefont{U}{msb}{m}{n}#1}} %gives us blackboard back -> http://tex.stackexchange.com/questions/214570/try-to-use-ams-blackboard-bold-font-together-with-texgyrepagella

\makeatletter
\g@addto@macro\bfseries{\boldmath}
\makeatother

\makeatletter
\g@addto@macro\mdseries{\unboldmath}
\g@addto@macro\normalfont{\unboldmath}
\g@addto@macro\rmfamily{\unboldmath}
\g@addto@macro\upshape{\unboldmath}
\makeatother

\DeclareCiteCommand{\citem}
    {}
    {\mkbibbrackets{\bibhyperref{\usebibmacro{postnote}}}}
    {\multicitedelim}
    {}

\renewcommand*{\multicitedelim}{\addcomma\space}

\AtEveryCitekey{\clearfield{note}}

\newcommand{\myhref}[1]{%
  \iffieldundef{doi}
    {\iffieldundef{url}
       {#1}
       {\href{\strfield{url}}{#1}}}
    {\href{http://dx.doi.org/\strfield{doi}}{#1}}%
}

\DeclareFieldFormat{title}{\myhref{\mkbibemph{#1}}}
\DeclareFieldFormat
  [article,inbook,incollection,inproceedings,patent,thesis,unpublished]
  {title}{\myhref{\mkbibquote{#1\isdot}}}

\makeatletter
\AtBeginDocument{%
    \newlength{\temp@x}%
    \newlength{\temp@y}%
    \newlength{\temp@w}%
    \newlength{\temp@h}%
    \def\my@coords#1#2#3#4{%
      \setlength{\temp@x}{#1}%
      \setlength{\temp@y}{#2}%
      \setlength{\temp@w}{#3}%
      \setlength{\temp@h}{#4}%
      \adjustlengths{}%
      \my@pdfliteral{\strip@pt\temp@x\space\strip@pt\temp@y\space\strip@pt\temp@w\space\strip@pt\temp@h\space re}}%
    \ifpdf
      \typeout{In PDF mode}%
      \def\my@pdfliteral#1{\pdfliteral page{#1}}% I don't know why % this command...
      \def\adjustlengths{}%
    \fi
    \ifxetex
      \def\my@pdfliteral #1{}% isn't equivalent to this one
      \def\adjustlengths{\setlength{\temp@h}{-\temp@h}\addtolength{\temp@y}{1in}\addtolength{\temp@x}{-1in}}%
    \fi%
    \def\Hy@colorlink#1{%
      \begingroup
        \ifHy@ocgcolorlinks
          \def\Hy@ocgcolor{#1}%
          \my@pdfliteral{q}%
          \my@pdfliteral{7 Tr}% Set text mode to clipping-only
        \else
          \HyColor@UseColor#1%
        \fi
    }%
    \def\Hy@endcolorlink{%
      \ifHy@ocgcolorlinks%
        \my@pdfliteral{/OC/OCPrint BDC}%
        \my@coords{0pt}{0pt}{\pdfpagewidth}{\pdfpageheight}%
        \my@pdfliteral{F}% Fill clipping path (the url's text) with
        \my@pdfliteral{EMC/OC/OCView BDC}%
        \begingroup%
          \expandafter\HyColor@UseColor\Hy@ocgcolor%
          \my@coords{0pt}{0pt}{\pdfpagewidth}{\pdfpageheight}%
          \my@pdfliteral{F}% Fill clipping path (the url's text)
        \endgroup%
        \my@pdfliteral{EMC}%
        \my@pdfliteral{0 Tr}% Reset text to normal mode
        \my@pdfliteral{Q}%
      \fi
      \endgroup
    }%
}
\makeatother

\colorlet{DarkRed}{red!50!black}
\colorlet{DarkGreen}{green!50!black}
\colorlet{DarkBlue}{blue!50!black}

\hypersetup{
	linkcolor = DarkRed,
	citecolor = DarkGreen,
	urlcolor = DarkBlue,
	bookmarksnumbered = true,
	linktocpage = true
}

\declaretheorem[numberwithin=section]{theorem}
\declaretheorem[numberlike=theorem]{lemma}

\declaretheorem[numberlike=theorem]{definition}
\declaretheorem[numberlike=theorem]{claim}
\declaretheorem[numberlike=theorem, style=remark]{remark}

\DontPrintSemicolon
\SetKw{KwAnd}{and}
\SetProcNameSty{textsc}
\SetFuncSty{textsc}

\newcommand{\dist}{d}
\newcommand{\s}{\ensuremath{s}\xspace}
\newcommand{\q}{q}

\title{Sublinear-Time Maintenance of Breadth-First Spanning Trees in Partially Dynamic Networks\thanks{This article appears in \emph{ACM Transactions on Algorithms} 13.4 (2017), pp 51:1--51:24, \href{http://dx.doi.org/10.1145/3146550}{doi:10.1145/3146550}. A preliminary version of this paper was presented at the \emph{40th International Colloquium on Automata, Languages and Programming (ICALP 2013)}.}}

\date{}

\author{
Monika Henzinger\thanks{University of Vienna, Faculty of Computer Science, Austria. Supported by the the University of Vienna (IK \mbox{I049-N}). The research leading to this work has received funding from the European Union's Seventh Framework Programme (FP7/2007-2013) under grant agreement no.\ 317532 and from the European Research Council under the European Union's Seventh Framework Programme (FP7/2007-2013)/ERC Grant Agreement no.\ 340506.}
\and Sebastian Krinninger\samethanks[2]
\and Danupon Nanongkai\thanks{KTH Royal Institute of Technology. Work partially done while at University of Vienna, Austria, ICERM, Brown University, USA, and Nanyang Technological University, Singapore 637371, and while supported in part by the following research grants: Nanyang Technological University grant M58110000, Singapore Ministry of Education (MOE) Academic Research Fund (AcRF) Tier 2 grant MOE2010-T2-2-082, and Singapore MOE AcRF Tier 1 grant MOE2012-T1-001-094.}
}

\hypersetup{
	pdftitle = {Sublinear-Time Maintenance of Breadth-First Spanning Trees in Partially Dynamic Networks},
	pdfauthor = {Monika Henzinger, Sebastian Krinninger, Danupon Nanongkai}
}

\begin{document}
\maketitle
\begin{abstract}
We study the problem of maintaining a {\em breadth-first spanning tree} (BFS tree) in {\em partially dynamic} distributed networks modeling a sequence of either failures or additions of communication links (but not both). We present deterministic $(1+\epsilon)$-approximation algorithms whose amortized time (over some number of link changes) is {\em sublinear} in $D$, the {\em maximum diameter} of the network.

Our technique also leads to a deterministic $(1+\epsilon)$-approximate incremental algorithm for single-source shortest paths (SSSP) in the sequential (usual RAM) model. Prior to our work, the state of the art was the classic \emph{exact} algorithm of Even and Shiloach~\cite{EvenS81} that is optimal under some assumptions \cite{RodittyZ11,HenzingerKNS15}. Our result is the first to show that, in the incremental setting, this bound can be beaten in certain cases if some approximation is allowed.

\end{abstract}
\newpage

%\tableofcontents
\newpage

\section{Introduction}

Complex networks are among the most ubiquitous models of interconnections between a multiplicity of individual entities, such as computers in a data center, human beings in society, and neurons in the human brain. The connections between these entities are constantly changing; new computers are gradually added to data centers, or humans regularly make new friends. These changes are usually {\em local} as they are known only to the entities involved. Despite their locality, they could affect the network {\em globally}; a single link failure could result in several routing path losses or destroy the network connectivity. To maintain its robustness, the network has to quickly respond to changes and repair its infrastructure. The study of such tasks has been the subject of several active areas of research, including dynamic, self-healing, and self-stabilizing networks.  

One important infrastructure in distributed networks is the {\em breadth-first spanning (BFS) tree} \cite{Lynch96,Peleg00}. It can be used, for instance, to approximate the network diameter and to provide a communication backbone for broadcast, routing, and control. In this paper, we study the problem of maintaining a BFS tree from a root node on dynamic distributed networks. Our main interest is repairing a BFS tree as fast as possible after each topology change.

\paragraph{Model.} We model the communication network by the CONGEST model \cite{Peleg00}, one of the major models of (locality-sensitive) distributed computation.
Consider a synchronous  network of processors  modeled by an undirected unweighted graph $ G = (V, E) $, where nodes model the processors and edges model the bounded-bandwidth links between the processors. We let $V$ and $E$ denote the set of nodes and edges of~$G$, respectively, and let $ \s $ be a specified \emph{root node}.
For any node $u$ and $v$, we denote by $\dist_G(u, v)$ the distance between $u$ and $v$ in $G$. 
The processors  (henceforth, nodes) are assumed to have unique IDs of $O(\log n)$ bits and infinite computational power. Each node has limited topological knowledge; in particular, it only knows the IDs of its neighbors and knows {\em no} other topological information (such as whether its neighbors are linked by an edge or not).
The communication is synchronous and occurs in discrete pulses, called {\em rounds}. All the nodes wake up simultaneously at the beginning of each round. In each round each node $u$ is allowed to send an arbitrary message of  $O(\log n)$ bits through each edge $ (u, v) $ that is adjacent to $u$, and the message will reach $v$ at the end of the current round. There are several measures to analyze the performance of such algorithms, a fundamental one being the {\em running time}, defined as the worst-case number of rounds of distributed communication. 

We model dynamic networks by a sequence of {\em attack} and {\em recovery} stages following the initial {\em preprocessing}. 
The dynamic network starts with a preprocessing on the initial network denoted by $G_0$, where 
nodes communicate on $G_0$ for some number of rounds. Once the preprocessing is finished, we begin the first attack stage where we assume that an adversary, who sees the current network $G_0$ and the states of all nodes, inserts and deletes an arbitrary number of edges in $G_0$. We denote the resulting network by $G_1$. This is followed by the first recovery stage where 
we allow nodes to communicate on $G_1$. After the nodes have finished communicating, the second attack stage starts, followed by the second recovery stage, and so on.
For any algorithm, we let the {\em total update time} be the total number of rounds needed by nodes to communicate during all recovery stages. Let the {\em amortized update time} be the total time divided by $q$ which is defined as the number of edges inserted and deleted. Important parameters in analyzing the running time are $n$, the number of nodes (which remains the same throughout all changes) and $D$, the {\em maximum diameter}, defined to be the maximum diameter among all networks in $\{G_0, G_1, \ldots\}$.
If some network $ G_t $ is not connected, we define its diameter as the diameter of the connected component containing the root node.
Note that $D\leq n$ according to this definition.
Following the convention from the area of (sequential) dynamic graph algorithms, we say that a dynamic network is {\em fully dynamic} if both insertions and deletions can occur in the attack stages.
Otherwise, it is {\em partially dynamic}. Specifically, if only edge insertions can occur,
it is an {\em incremental dynamic network}. If only edge deletions can occur, it is {\em decremental}.

Our model highlights two aspects of dynamic networks: (1) How quickly a network can recover its infrastructure after changes and (2) how edge failures and additions affect the network. These aspects have been studied earlier but we are not aware of any previous model identical to ours. To highlight these aspects, a few assumptions are inherent in our model. 
First, it is assumed that the network remains static in each recovery stage. This assumption is often used (e.g., \cite{Korman08,HayesST12,KrizancLR04,MalpaniWV00}) and helps to emphasize the running time aspect of dynamic networks.
Also note that we assume that the network is synchronous, but our algorithms will also work in an asynchronous model under the same asymptotic time bounds, using a synchronizer \cite{Peleg00,Awerbuch85}. 
Furthermore, we consider amortized update time which is similar in spirit to the amortized communication complexity heavily studied earlier (e.g., \cite{AwerbuchCK08}).
Finally, the results in this paper are on partially dynamic networks. While fully dynamic algorithms are more desirable, we believe that the partially dynamic setting is worth studying, for two reasons. The first reason, which is our main motivation, comes from an experience in the study of sequential dynamic algorithms,
where insights from the partially dynamic setting often lead to improved fully dynamic algorithms. 
Moreover, partially dynamic algorithms can be useful in cases where one type of changes occurs much more frequently than the other type. 
For example, links constantly fail in physical networks, and it might not be necessary that the network has to be fixed (by adding a link) immediately. Instead, the network can try to maintain its infrastructures under a sequence of failures until the quality of service cannot be guaranteed anymore, e.g., the network diameter becomes too large. Partially dynamic algorithms for maintaining a BFS tree, which in turn maintains the approximate network diameter, are quite suitable for this type of applications.

\paragraph{Problem.} We are interested in maintaining an approximate BFS tree.
Our definition of approximate BFS trees below is a modification of the definition of BFS trees in \cite[Definition 3.2.2]{Peleg00}. 
\begin{definition}[Approximate BFS tree] 
For any $\alpha\geq 1$, an {\em $\alpha$-approximate BFS tree} of an unweighted undirected graph $G$ with respect to a given root $s$ is a spanning tree~$T$ of the connected component containing $ s $ such that for every node $v$ connected to $ \s $, $\dist_T(v, \s)\leq \alpha \dist_{G}(v, \s)$.
If $ \alpha = 1 $, then $ T $ is an \emph{(exact) BFS tree}.
\end{definition}
Note that, for any spanning tree $ T $ of $ G $, $\dist_T(v, \s)\geq \dist_{G}(v, \s)$.
Our goal is to maintain an approximate BFS tree $T_t$ at the end of each recovery stage~$t$ in the sense that every node $v$ knows its approximate distance to the preconfigured root $s$ in $G_t$ and, for each neighbor $ u $ of $ v $, $v$ knows if $u$ is its parent or child in $T_t$.
Note that for convenience we will usually consider $ \dist_{G} (v, \s) $, the distance of $ v $ \emph{to} the root, instead of $ \dist_{G} (\s, v) $, the distance of $ v $ \emph{from} the root.
In an undirected graph both values are the same.

\paragraph{Naive Algorithm.} As a toy example, observe that we can maintain a BFS tree simply by recomputing a BFS tree from scratch in each recovery stage. By using the standard algorithm (see, e.g., \cite{Peleg00,Lynch96}), we can do this in time $O(D_t)$, where $D_t$ is the diameter of the graph~$G_t$. Thus, the update time is $O(D)$.

\paragraph{Results.}
Our main results are partially dynamic algorithms that break the naive update time of $ O (D) $ in the long term.
They can maintain, for any constant $ 0 < \epsilon \leq 1 $, a $(1+\epsilon)$-approximate BFS tree in time that is {\em sublinear in $D$} when amortized over $\omega( n / D)$ edge changes. 
To be precise, the amortized update time over $q$ edge changes is 
\begin{equation*}
O \left( \frac{n^{1/3} D^{2/3}}{ \epsilon^{2/3}q^{1/3}} \right) ~~~\mbox{and}~~~ O \left( \frac{n^{1/5} D^{4/5}}{\epsilon q^{1/5}} \right)
\end{equation*}
in the incremental and decremental setting, respectively.
For the particular case of $q=\Omega(n)$, we get amortized update times of $O(D^{2/3} / \epsilon^{2/3})$ and $O(D^{4/5} / \epsilon)$ for the incremental and decremental cases, respectively.
Our algorithms do not require any prior knowledge about the dynamic network, e.g., $D$ and $q$.
We have formulated the algorithms for a setting that allows insertions or deletions of edges.
The guarantees of our algorithms also hold when we allow insertions or deletions of \emph{nodes}, where the insertion of a node also inserts all its incident edges and the deletion of a node also deletes all its incident edges.
In the running time, the parameter $ q $ then counts the number of node insertions or node deletions, respectively.

We note that, while there is no previous literature on this problem, one can parallelize the algorithm of Even and Shiloach~\cite{EvenS81} (see also \cite{King99,RodittyZ11}) to obtain an amortized update time of $O(nD/q + 1)$ over $q$ changes in both the incremental and the decremental setting. This bound is sublinear in $D$ when $q=\omega(n)$. 
Our algorithms give a sublinear time guarantee for a smaller number of changes, especially in applications where $D$ is large.
They are faster than the Even-Shiloach algorithm when $ q = \omega (\epsilon n \sqrt{D}) $ (incremental) and $ q = \omega (\epsilon^{7/12} n D^{1/6}) $ (decremental).

In the sequential (usual RAM) model, our technique also gives an $(1+\epsilon)$-approximation algorithm for the incremental single-source shortest paths (SSSP) problem with an amortized update time of $O(mn^{1/4} \log{n} / \sqrt{\epsilon q})$ 
per insertion and $O(1)$ query time, where $m$ is the number of edges in the final graph, and $q$ is the number of edge insertions. Prior to this result, only the classic exact algorithm of Even and Shiloach \cite{EvenS81} from the 80s, with $O(mn/q)$ amortized update time, was known.
No further progress has been made in the last three decades.
Roditty and Zwick~\cite{RodittyZ11} provided an explanation for this by showing that the algorithm of Even and Shiloach~\cite{EvenS81} is likely to be the fastest combinatorial \emph{exact} algorithm, assuming that there is no faster combinatorial algorithm for Boolean matrix multiplication. 
More recently Henzinger et al.~\cite{HenzingerKNS15} showed that by assuming a different conjecture, called Online Matrix-Vector Multiplication Conjecture, this statement can be extended to any algorithm (including non-combinatorial ones).
Bernstein and Roditty~\cite{BernsteinR11} showed that, in the decremental setting, this bound can be broken if some approximation is allowed. Our result is the first one of the same spirit in the \emph{incremental} setting for deterministic algorithms; i.e., we break the bound of Even and Shiloach for the case $q= o(n^{3/2})$, which in particular applies when $m= o(n^{3/2})$.
The techniques introduced in this paper (first presented in the preliminary version \cite{HenzingerKN-ICALP13}) together with techniques from~\cite{HenzingerKN-SICOMP15} also led to a decremental algorithm \cite{HenzingerKN-SODA14} that improves the result of Bernstein and Roditty~\cite{BernsteinR11}.
We finally obtained a near-optimal algorithm in the decremental setting \cite{HenzingerKN-FOCS14}, which is a significant improvement over \cite{BernsteinR11}.
In terms of deterministic algorithms, Bernstein and Chechik have recently presented improved incremental and decremental algorithms for dense~\cite{BernsteinC16} and sparse graphs~\cite{BernsteinC17}.
For very sparse graphs with $m = \Theta (n)$, the incremental algorithm in this paper still remains the fastest.
%We note that the algorithm in this paper is still the fastest deterministic algorithm when $q= o(n^{3/2})$ (and thus when $m= o(n^{3/2})$). In fact, there is no deterministic algorithm faster than Even and Shiloach's algorithm, even for some range of parameters, except this one.

\paragraph{Related Work.} The problem of computing on dynamic networks is a classic problem in the area of distributed computing, studied from as early as the 70s; see, e.g., \cite{AwerbuchCK08} and references therein. The main motivation is that dynamic networks better capture real networks, which experience failures and additions of new links.
There is a large number of models of dynamic networks in the literature, each emphasizing different aspects of the problem. Our model closely follows the model of the sequential setting and, as discussed earlier, highlights the amortized update time aspect. 
It is closely related to the model in \cite{KormanP08} where the main goal is to optimize the amortized update time using static algorithms in the recovery stages. The model in \cite{KormanP08} is still slightly different from ours in terms of allowed changes. For example, the model in \cite{KormanP08} considers weighted networks and allows small weight changes but no topological changes; moreover, the message size can be unbounded (i.e., the static algorithm in the recovery stage operates under the so-called LOCAL model).
Another related model the {\em controlled dynamic model} (e.g., \cite{KormanK13,AfekAPS96}) where the topological changes do not happen instantaneously but are delayed until getting a permit to do so from the resource controller. Our algorithms can be used in this model as well since we can delay the changes until each recovery stage is finished.
Our model is similar to, and can be thought of as a combination of, two types of models: those in, e.g., \cite{Korman08,HayesST12,KrizancLR04,MalpaniWV00} whose main interest is to determine how fast a network can recover from changes using static algorithms in the recovery stages, and those in, e.g., \cite{AwerbuchCK08,AfekAG87,Elkin07}, which focus on the amortized cost per edge change. 
Variations of partially dynamic distributed networks have also been considered (e.g., \cite{Italiano91,RamaraoV92,CiceroneDSFP07,CiceroneDSF10}).

The problem of constructing a BFS tree has been studied intensively in various distributed settings for decades (see \cite[Chapter 5]{Peleg00}, \cite[Chapter 4]{Lynch96} and references therein). The studies were also extended to more sophisticated structures such as minimum spanning trees (e.g., \cite{GarayKP93,KuttenP98,PelegR00,Elkin06,LotkerPP06,LotkerPPP05,KorKP13,DasSarmaHKKNPPW12,ElkinKNP14}) and Steiner trees \cite{KhanKMPT12}. 
These studies usually focus on {\em static} networks, i.e., they assume that the network never changes and want to construct a BFS tree once, from scratch. 
While we are not aware of any results on maintaining a BFS tree on dynamic networks, there are a few related results. 
Much attention (e.g., \cite{AwerbuchCK08}) has previously been given to the problem of {\em maintaining a spanning tree}. In a seminal paper by Awerbuch~\cite{AwerbuchCK08}, it was shown that the amortized message complexity of maintaining a spanning tree can be significantly smaller than the cost of the previous approach of recomputing from scratch \cite{AfekAG87}.\footnote{A variant of their algorithm was later implemented as a part of the PARIS networking project at IBM \cite{CidonGKK95} and slightly improved \cite{KuttenP99}.}
Our result is in the same spirit as \cite{AwerbuchCK08} in breaking the cost of recomputing from scratch. 
An attempt to maintain spanning trees of small diameter has also motivated a problem called {\em best swap}. The goal is to replace a failed edge in the spanning tree by a new edge in such a way that the diameter is minimized. This problem has recently gained considerable attention in both sequential (e.g., \cite{AlstrupHLT05,ItalianoR98,NardelliPW01,NardelliPW03,SalvoP07,ItoIOY05,DasGW10,Gfeller12}) and distributed (e.g., \cite{GfellerSW11,FlocchiniEPPS06}) settings.

In the sequential dynamic graph algorithms literature, a problem similar to ours is the single-source shortest paths (SSSP) problem on undirected graphs. 
This problem has been studied in partially dynamic settings and has applications to other problems, such as all-pairs shortest paths and reachability.
As we have mentioned earlier, the classic bound of \cite{EvenS81}, which might be optimal \cite{RodittyZ11,HenzingerKNS15}, has recently been improved by randomized decremental approximation algorithms \cite{BernsteinR11,HenzingerKN-SODA14,HenzingerKN-FOCS14}, and we achieve a similar result in the incremental setting with a deterministic algorithm. 
Since our algorithms use the algorithm of \cite{EvenS81} as a subroutine, we formally state its guarantees in the following.
As mentioned above, this algorithm has not been considered in the distributed model before, but its analysis from the sequential model immediately carries over to the distributed model.\footnote{In the sequential model, the algorithm has to perform work proportional to the degree of each node whose distance to the root decreases (increases). Assume we are interested in a shortest paths tree up to depth $ X $. As each node's distance to the root can increase (decrease) at most $ X $ times, the total running time is $ O (m X) $. In the distributed model, sending a message to all neighbors takes one round and thus we only charge constant time to each level increase (decrease) of a node, resulting in a total time of $ O (n \min (X, D) + q) $. The additional $ q $ comes from the fact that we have to spend constant time per insertion (deletion), which in the sequential model is dominated by other running time aspects.} Since we will need this result later in this paper, we state it here. 

\begin{theorem}[\cite{EvenS81}]\label{thm:Even_Shiloach}
There is a partially dynamic algorithm for maintaining a shortest paths tree from a given root node up to depth $ X \leq n $ under edge insertions (deletions) in an unweighted, undirected graph.
Its total running time over $ q $ insertions (deletions) is $ O (mX) $ in the sequential model and $ O (n \min (X, D) + q) $ in the distributed model.
\end{theorem}

\section{Main Technical Idea} \label{sec:idea}

All our algorithms are based on a simple idea of modifying Even-Shiloach algorithm~\cite{EvenS81} with lazy updates, which we call \emph{lazy Even-Shiloach tree}. Implementing this idea on different models requires modifications to cope with difficulties and to maximize efficiency. 
In this section, we explain the main idea by sketching a simple algorithm and its analysis for the incremental setting in the sequential and the distributed model. 
We start with an algorithm that has {\em additive error}: Let $\kappa$ and $\delta$ be parameters. For every recovery stage $t$, we maintain a tree $T_t$ such that $\dist_{G_t}(v, \s)\leq \dist_{T_t}(v, \s)\leq \dist_{G_t}(v, \s)+\kappa \delta$ for every node $ v $.
We will do this by recomputing a BFS tree from scratch repeatedly, specifically $O(q/\kappa+nD/\delta^2)$ times during $ q $ updates.

During the preprocessing, our algorithm constructs a BFS tree of $G_0$, denoted by $T_0$. This means that every node $u$ knows its parent and children in $T_0$ and the value of $\dist_{T_0}(u, \s)$. Suppose that, in the first attack stage, an edge is inserted, say $(u, v)$ where $\dist_{G_0}(u, \s) > \dist_{G_0}(v, \s)$. As a result, the distance from $u$ to $ \s $ might decrease, i.e. $\dist_{G_1}(u, \s)<\dist_{G_0}(u, \s)$.
In this case, the distances from $s$ to some other nodes (e.g., the children of $u$ in $T_0$) could decrease as well, and we may wish to recompute the BFS tree. Our approach is to do this {\em lazily}: We recompute the BFS tree only when the distance from $u$ to $s$ decreases by at least $\delta$; otherwise, we simply do nothing! In the latter case, we say that {\em $u$ is lazy}.
Additionally, we regularly ``clean up'' by recomputing the BFS tree after each $\kappa$ insertions.

To prove an additive error of $\kappa\delta$, observe that errors occur for this single insertion only when $v$ is lazy. 
Intuitively, this causes an additive error of $\delta$ since we could have decreased the distance of $v$ and other nodes by at most $\delta$, but we did not. This argument can be extended to show that if we have $i$ lazy nodes, then the additive error will be at most $i\delta$. Since we do the cleanup each $\kappa$ insertions, the additive error will be at most $\kappa\delta$ as claimed. 

To bound the number of BFS tree recomputations, first observe that the cleanup clearly contributes $O(q/\kappa)$ recomputations in total, over $q$ insertions. 
Moreover, a recomputation could also be caused by some node $v$, whose distance to $s$ decreases by at least $\delta$.
Since every time a node $v$ causes a recomputation, its distance decreases by at least $\delta$, and since $\dist_{G_0}(v, \s)\leq D$, $v$ will cause the recomputation at most $D/\delta$ times. This naive argument shows that there are $nD/\delta$ recomputations (caused by $n$ different nodes) in total. 
This analysis is, however, {\em not} enough for our purpose. A tighter analysis, which is crucial to all our algorithms relies on the observation that when $v$ causes a recomputation, the distance from any neighbor of $v$, say $ v' $, to $s$ also decreases by at least $\delta-1$. Similarly, the distance of any neighbor of $v'$ to $s$ decreases by at least $\delta-2$, and so on. This leads to the conclusion that one recomputation corresponds to $(\delta+(\delta-1)+(\delta-2)+\ldots)=\Omega(\delta^2)$ distance decreases. Thus, the number of recomputations is at most $nD/\delta^2$.
Combining the two bounds, we get that the number of BFS tree computations is $O(q/\kappa+nD/\delta^2)$ as claimed above.
We get a bound on the total time when we multiply this number by the time needed for a single BFS tree computation.
In the sequential model this takes time $ O (m) $, where $ m $ is the final number of edges, and in the distributed model this takes time $ O (D) $, where $ D $ is the dynamic diameter of the network.

To convert the additive error into a multiplicative error of $(1+\epsilon)$, we execute the above algorithm only for nodes whose distances to $s$ are greater than $\kappa\delta/\epsilon$. For other nodes, we can use the algorithm of Even and Shiloach~\cite{EvenS81} to maintain a BFS tree of depth $\kappa\delta/\epsilon$. This requires an additional time of $O(m\kappa\delta/\epsilon)$ in the sequential model and $O(n\kappa\delta/\epsilon)$ in the distributed model. 

By setting $\kappa$ and $\delta$ appropriately, the above incremental algorithm immediately gives total update times of $ O (m n^{2/5} q^{2/5} / \epsilon^{2/5}) $ and $ O (q^{2/5} n^{3/5} D^{4/5} / \epsilon^{2/5}) $ in the sequential and distributed model, respectively.
To obtain the running time bounds claimed in the introduction of this paper, we need one more idea called {\em layering}, where we use different values of $\delta$ and $\kappa$ depending on the distance of each node to $ s $.
In the decremental setting, the situation is much more difficult, mainly because it is expensive for a node $v$ to determine how much its
distance to $s$ has increased after a deletion. Moreover, unlike the incremental case, nodes cannot simply ``do nothing'' when an edge is deleted. We have to cope with this using several other ideas, e.g., constructing an virtual tree (in which edges sometimes represent paths).

\section{Incremental Algorithm}\label{sec:incremental_algorithm}

In this section we present a framework for an incremental algorithm that allows up to $ q $ edge insertions and provides an additive approximation of the distances to a distinguished node $ \s $.
Subsequently we will explain how to use this algorithm to get $ (1+\epsilon) $-approximations in the sequential model and the distributed model, respectively.
For simplicity we assume that the initial graph is connected.
In Section~\ref{sec:removing_connectedness_assumption_incremental} we explain how to remove this assumption.

\subsection{General Framework}\label{sec:incremental_abstract}

The algorithm (see Algorithm~\ref{alg:insertion_algorithm}) works in {\em phases}.
At the beginning of every phase we compute a BFS tree $ T_0 $ of the current graph, say $ G_0 $.
Every time an edge $ (u, v) $ is inserted, the distances of some nodes to $ \s $ in $ G $ might decrease.
Our algorithm tries to be {\em as lazy as possible}. That is, when the decrease does not exceed some parameter~$ \delta $, our algorithm keeps its tree $ T_0 $ and accepts an {\em additive error} of $ \delta $ for every node.
When the decrease exceeds $ \delta $, our algorithm starts a new phase and recomputes the BFS tree.
It also starts a new phase after each $ \kappa $ edge insertions to keep the additive error limited to $ \kappa \delta $.
The algorithm will answer a query for the distance from a node $ x $ to $ \s $ by returning $ \dist_{G_0} (x, \s) $, the distance from $ x $ to $ \s $ at the beginning of the current phase.
It can also return the path from $ x $ to $ \s $ in $ T_0 $ of length $ \dist_{G_0} (x, \s) $.
Besides~$ \delta $ and~$ \kappa $, the algorithm has a third parameter $ X $ which indicates up to which distance from~$ \s $ the BFS tree will be computed.
In the following we denote by $ G_0 $ the state of the graph at the beginning of the current phase and by $ G $ we denote the current state of the graph after all insertions so far.

\begin{algorithm}
\caption{Incremental algorithm}
\label{alg:insertion_algorithm}

\SetKwProg{procedure}{Procedure}{}{}
\SetKwFunction{initialize}{Initialize}
\SetKwFunction{insert}{Insert}

\procedure{\insert{$u$, $v$}}{
	$ k \gets k+1 $\;
	\lIf{$ k = \kappa $}{
		\initialize{}
	}
	\lIf{$ \dist_{G_0} (u, \s) > \dist_{G_0} (v, \s) + \delta $}{\label{lin:main_rule_incremental_algorithm}
		\initialize{}
	}
}

\BlankLine

\procedure(\tcp*[f]{Start new phase}){\initialize{}}{
	$ k \gets 0 $\;
	Compute BFS tree $ T $ of depth $ X $ rooted at $ \s $ and current distances $ \dist_{G_0} (\cdot, \s) $\;
}
\end{algorithm}

As we show below the algorithm gives the desired additive approximation by considering the shortest path of a node $ x $ to the root $ \s $ in the current graph $ G $.
By the main rule in Line~\ref{lin:main_rule_incremental_algorithm} of the algorithm, the inequality $ \dist_{G_0} (u, \s) \leq \dist_{G_0} (v, \s) + \delta $ holds for every edge $ (u, v) $ that was inserted since the beginning of the current phase (otherwise a new phase would have been started).
Since at most $ \kappa $ edges have been inserted, the additive error is at most $ \kappa \delta $.
\begin{lemma}[Additive Approximation]\label{lem:incremental_algorithm_main_invariant}
For every $ \kappa \geq 1 $ and $ \delta \geq 1 $, Algorithm~\ref{alg:insertion_algorithm} provides the following approximation guarantee for every node~$ x $ such that ${ \dist_{G_0} (x, \s) \leq X }$:
\begin{equation*}
\dist_G (x, \s) \leq \dist_{G_0} (x, \s) \leq \dist_G (x, \s) + \kappa \delta \, .
\end{equation*}
\end{lemma}

\begin{proof}
The algorithm can only provide the approximation guarantee for every node~$ x $ such that $ \dist_{G_0} (x, \s) \leq X $ because other nodes are not contained in the BFS tree of the current phase.
It is clear that $ \dist_G (x, \s) \leq \dist_{G_0} (x, \s) $ because $ G $ is the result of inserting edges into $ G_0 $.
In the following we argue about the second inequality.

Consider the shortest path $ \pi = x_l, x_{l-1}, \ldots x_0 $ of length $ l $ from $ x $ to $ \s $ in $ G $ (where $ x_l = x $ and $ x_0 = \s $).
Let $ S_j $ (with $ 0 \leq j \leq l $) denote the number of edges in the subpath $ x_j, x_{j-1}, \ldots, x_0 $ that were inserted since the beginning of the current phase.
\begin{claim}
For every integer $ j $ with $ 0 \leq j \leq l $ we have $ \dist_{G_0} (x_j, \s) \leq \dist_G (x_j, \s) + S_j \delta $.
\end{claim}

Clearly the claim already implies the inequality we want to prove since there are at most $ \kappa $ edges that have been inserted since the beginning of the current phase which gives the following chain of inequalities:
\begin{equation*}
\dist_{G_0} (x, \s) = \dist_{G_0} (x_l, \s) \leq \dist_G (x_l, \s) + S_l \delta \leq \dist_G (x, \s) + \kappa \delta \, .
\end{equation*}

Now we proceed with the inductive proof of the claim
The induction base $ j = 0 $ is trivially true because $ x_j = \s $.
Now consider the induction step where we assume that the inequality holds for $ j $ and we have to show that it also holds for $ j+1 $.

Consider first the case that the edge $ (x_{j+1}, x_j) $ is one of the edges that have been inserted since the beginning of the current phase.
By the rule of the algorithm we know that $ \dist_{G_0} (x_{j+1}, \s) \leq \dist_{G_0} (x_j, \s) + \delta $ and by the induction hypothesis we have $ \dist_{G_0} (x_j, \s) \leq \dist_G (x_j, \s) + S_j \delta $.
By combining these two inequalities we get $ \dist_{G_0} (x_{j+1}, \s) \leq \dist_G (x_j, \s) + (S_j + 1) \delta $.
The desired inequality now follows because $ S_{j+1} = S_j + 1 $ and because $ \dist_G (x_j, \s) \leq \dist_G (x_{j+1}, \s) $ (on the shortest path $ \pi $, $ x_j $ is closer to $ \s $ than $ x_{j+1} $).

Now consider the case that the edge $ (x_{j+1}, x_j) $ is not one of the edges that have been inserted since the beginning of the current phase.
In that case the edge $ (x_{j+1}, x_j) $ in contained in the graph $ G_0 $ and thus $ \dist_{G_0} (x_{j+1}, \s) \leq \dist_{G_0} (x_j, \s) + 1 $.
By the induction hypothesis we have $ \dist_{G_0} (x_j, \s) \leq \dist_G (x_j, \s) + S_j \delta $.
By combining these two inequalities we get $ \dist_{G_0} (x_{j+1}, \s) \leq \dist_G (x_j, \s) + 1 + S_j \delta $.
Since $ x_{j+1} $ and $ x_j $ are neighbours on the shortest path $ \pi $ in $ G $, we have $ \dist_G (x_{j+1}, \s) = \dist_G (x_j, \s) + 1 $.
Therefore we get $ \dist_{G_0} (x_{j+1}, \s) \leq \dist_G (x_{j+1}, \s) + S_j \delta $.
Since $ S_{j+1} = S_j $, the desired inequality follows.
\end{proof}

\begin{remark}
In the proof of Lemma~\ref{lem:incremental_algorithm_main_invariant} we need the property that at most $ \kappa $ edges on the shortest path to the root have been inserted since the beginning of the current phase.
If we allow inserting $ \kappa / 2 $ nodes (together with their set of incident edges) we will see at most $ \kappa $ inserted edges on the shortest path to the root as each node appears at most once on this path and contributes at most $ 2 $ incident edges.
Thus, we can easily modify our algorithms to deal with node insertions with the same approximation guarantee and asymptotic running time.
\end{remark}

If an edge $ (u, v) $ is inserted into the graph such that the inequality $ \dist_{G_0} (u, \s) \leq \dist_{G_0} (v, \s) + \delta $ does not hold (and subsequently the algorithm calls the procedure \initialize), we cannot guarantee our bound on the additive error anymore. Nevertheless the algorithm makes progress in some sense: After the insertion, $ u $ has an edge to $ v $ whose initial distance to $ \s $ was significantly smaller than the one from $ u $ to $ s $.
This implies that the distance from $ u $ to $ \s $ has decreased by at least $ \delta $ since the beginning of the current phase.
Thus testing whether $ \dist_{G_0} (u, \s) > \dist_{G_0} (v, \s) + \delta $ is a fast way of testing whether $ \dist_{G_0} (u, \s) \geq \dist_G (u, \s) + \delta $, i.e., whether the distance between $ u $ and $ \s $ has decreased so much that a rebuild is necessary.
\begin{lemma}\label{lem:incremental_algorithm_distance_increase}
If an edge $ (u, v) $ is inserted such that $ \dist_{G_0} (u, \s) > \dist_{G_0} (v, \s) + \delta $, then $ \dist_{G_0} (u, \s) \geq \dist_G (u, \s) + \delta $.
\end{lemma}

\begin{proof}
We have inserted an edge $ (u, v) $ such that $ \dist_{G_0} (u, \s) > \dist_{G_0} (v, \s)  + \delta $ (which is equivalent to $ \dist_{G_0} (v, \s) \leq \dist_{G_0} (u, \s) - \delta - 1 $).
In the current graph $ G $, we already have inserted the edge $ (u, v) $ and therefore $ \dist_G (u, \s) \leq \dist_G (v, \s) + 1 $.
Since $ G $ is the result of inserting edges into $ G_0 $, distances in $ G $ are not longer than in $ G_0 $, and in particular $ \dist_G (v, \s) \leq \dist_{G_0} (v, \s) $.
Therefore we arrive at the following chain of inequalities:
\begin{equation*}
\dist_G (u, \s) \leq \dist_G (v, \s) + 1 \leq \dist_{G_0} (v, \s) + 1 \leq \dist_{G_0} (u, \s) - \delta - 1 + 1 = \dist_{G_0} (u, \s) - \delta
\end{equation*}
Thus, we get $ \dist_{G_0} (u, \s) \geq \dist_G (u, \s) + \delta $.
\end{proof}

Since we consider undirected, unweighted graphs, a large decrease in distance for one node also implies a large decrease in distance for many other nodes.
\begin{lemma}\label{lem:distance_increase}
Let $ H = (V, E) $ and $ H' = (V, E') $ be unweighted, undirected graphs such that $ H $ is connected and $ E \subseteq E' $.
If there is a node $ y \in V $ such that $ \dist_H (y, \s) \geq \dist_{H'} (y, \s) + \delta $, then
$ \sum_{x \in V} \dist_H (x, \s) \geq \sum_{x \in V'} \dist_{H'} (x, \s) + \Omega (\delta^2) $.
\end{lemma}

\begin{proof}
Let $ \pi $ denote the shortest path from $ y $ to $ s $ of length $ \dist_{H} (y, s) $ in $ H $.
We first bound the distance change for single nodes.
\begin{claim}
For every node $ x $ on $ \pi $ we have $ \dist_{H} (x, s) \geq \dist_{H'} (x, s) + \delta - \dist_H (x, y) - \dist_{H'} (x, y) $.
\end{claim}
\begin{proof}[of Claim]
By the triangle inequality we have $ \dist_{H'} (x, \s) \leq \dist_{H'} (x, y) + \dist_{H'} (y, \s) $, which is equivalent to $ \dist_{H'} (y, \s) \geq \dist_{H'} (x, \s) - \dist_{H'} (x, y) $.
By this inequality and the fact that $ x $ lies on $ \pi $, the shortest path from $ y $ to $ \s $ in $ H $, we have
\begin{equation*}
\dist_H (y, x) + \dist_H (x, \s) = \dist_H (y, \s) \geq \dist_{H'} (y, \s) + \delta \geq \dist_{H'} (x, \s) - \dist_{H'} (x, y) + \delta \, .
\end{equation*}
Since $ \dist_H (y, x) = \dist_H (x, y) $ the claimed inequality follows.
\end{proof}
From the claim and the fact that $ \dist_{H'} (x, y) \leq \dist_H (x, y) $ we conclude that
\begin{align*}
\sum_{x \in \pi, \dist_{H}(x, y) < \delta/2} \dist_{H} (x, s) &\geq \sum_{x \in \pi, \dist_{H}(x, y) < \delta/2} (\dist_{H'} (x, s) + \delta - 2 \dist_{H} (x, y)) \\
  &= \left( \sum_{x \in \pi, \dist_{H}(x, y) < \delta/2} \dist_{H'} (x, s) \right) + \left( \sum_{x \in \pi, \dist_{H}(x, y) < \delta/2} (\delta - 2 \dist_{H} (x, y)) \right) \\
  &\geq \left( \sum_{x \in \pi, \dist_{H}(x, y) < \delta/2} \dist_{H'} (x, s) \right) + \left( \sum_{j=1}^{\lfloor \delta/2 \rfloor} (\delta - 2j) \right) \\
  &= \left( \sum_{x \in \pi, \dist_{H}(x, y) < \delta/2} \dist_{H'} (x, s) \right) + \delta (\lfloor \delta/2 \rfloor) - \lfloor \delta/2 \rfloor (\lfloor \delta/2 \rfloor + 1) \\
  &= \left( \sum_{x \in \pi, \dist_{H}(x, y) < \delta/2} \dist_{H'} (x, s) \right) + \Omega (\delta^2) \, .
\end{align*}

Finally, we get:
\begin{align*}
\sum_{x \in V} \dist_{H} (x, s) &= \left( \sum_{x \in \pi, \dist_{H}(x, y) < \delta/2} \dist_{H} (x, s) \right) + \sum_{x \notin \pi \text{ or } \dist_{H}(x, y) \geq \delta/2} \underbrace{\dist_{H} (x, s)}_{\geq \dist_{H'} (x, s)} \\
  &\geq \left( \sum_{x \in \pi, \dist_{H}(x, y) < \delta/2} \dist_{H'} (x, s) \right) + \Omega (\delta^2) + \sum_{x \notin \pi \text{ or } \dist_{H}(x, y) \geq \delta/2} \dist_{H'} (x, s) \\
  &= \left( \sum_{x \in V} \dist_{H'} (x, s) \right) + \Omega (\delta^2) \, . \qedhere
\end{align*}
\end{proof}

The quadratic distance decrease is the key observation for the efficiency of our algorithm as it limits the number of times a new phase starts, which is the expensive part of our algorithm.
\begin{lemma}[Running Time]\label{lem:running_time_incremental_abstract}
For every $ \kappa \geq 1 $ and $ \delta \geq 1 $, the total update time of Algorithm~\ref{alg:insertion_algorithm} is
$ O ( T_{\mathrm{BFS}} (X) \cdot (q / \kappa +  n X / \delta^2 + 1) + q) $,
where $ T_{\mathrm{BFS}} (X) $ is an upper bound on the time needed for computing a BFS tree up to depth~$ X $.
\end{lemma}

\begin{proof}
Besides the constant time per insertion we have to compute a BFS tree of depth $ X $ at the beginning of every phase.
The first cause for starting a new phase is that the number of edge deletions in a phase reaches $ \kappa $, which can happen at most $ q / \kappa $ times.
The second cause for starting a new phase is that we insert an edge $ (u, v) $ such that $ \dist_{G_0} (u, \s) > \dist_{G_0} (v, \s) + \delta $.
By Lemmas~\ref{lem:incremental_algorithm_distance_increase} and~\ref{lem:distance_increase} this implies that the sum of the distances of all nodes to $ \s $ has increased by at least $ \Omega (\delta^2) $ since the beginning of the current phase.
There are at most $ n $ nodes of distance at most $ X $ to $ s $ which means that the sum of the distances is at most $ n X $.
Therefore such a decrease can occur at most $ O (n X / \delta^2) $ times.
The overall running time thus is $ O ( T_{\mathrm{BFS}} (X) \cdot (q / \kappa + n X / \delta^2 + 1) + q) $.
The $1$-term is just a technical necessity as the BFS tree has to be computed at least once.
\end{proof}

The algorithm above provides an additive approximation.
In the following we turn this into a multiplicative approximation for a fixed distance range.
Using a multi-layer approach we enhance this to a multiplicative approximation for the full distance range in Sections~\ref{sec:incremental_sequential_model} (sequential model) and~\ref{sec:incremental_distributed_model} (distributed model).

\begin{lemma}[Multiplicative Approximation]\label{lem:multiplicative_approximation_abstract}
Let $ 0 < \epsilon \leq 1 $, $ X \leq n $, and set $ \gamma = \epsilon / 4 $.
If $ \gamma^2 q X \geq n $ and $ \gamma n X^2 \geq q $, then by setting
$ \kappa = q^{1/3} X^{1/3} \gamma^{2/3} / n^{1/3} $
and
$ \delta = n^{1/3} X^{2/3} \gamma^{1/3} / q^{1/3} $,
Algorithm~\ref{alg:insertion_algorithm} has a total update time of
\begin{equation*}
O \left(T_{\mathrm{BFS}} (X) \cdot \frac{q^{2/3} n^{1/3}}{\epsilon^{2/3} X^{1/3}} + q \right) \, ,
\end{equation*}
where $ T_{\mathrm{BFS}} (X) $ is an upper bound on the time needed for computing a BFS tree up to depth~$ X $.
Furthermore, it provides the following approximation guarantee:
$ \dist_{G_0} \geq \dist_G (x, \s) $ for every node $ x $ and $ \dist_{G_0} (x, \s) \leq (1 + \epsilon) \dist_G (x, \s) $ for every node $ x $ such that $ X/2 \leq \dist_{G_0} (x, \s) \leq X $.
\end{lemma}

\begin{proof}
To simplify the notation a bit we define $ A = \kappa \delta $, which gives $ A = \gamma X $.
By Lemma~\ref{lem:running_time_incremental_abstract}, Algorithm~\ref{alg:insertion_algorithm} runs in time
\begin{equation*}
O \left( T_{\mathrm{BFS}} (X) \cdot \left( \frac{q}{\kappa} + \frac{n X}{\delta^2} + 1 \right) + q \right) \, .
\end{equation*}
It is easy to check that by our choices of $ \kappa $ and $ \delta $ the two terms appearing in the running time are balanced and we get
\begin{equation*}
\frac{q}{\kappa} = \frac{n X}{\delta^2} = \frac{ q^{2/3} n^{1/3}}{\gamma^{2/3} X^{1/3}} = O \left(\frac{ q^{2/3} n^{1/3}}{\epsilon^{2/3} X^{1/3}} \right) \, .
\end{equation*}
Furthermore the inequalities $ \gamma^2 q X \geq n $ and $ \gamma n X^2 \geq q $ ensure that $ \kappa \geq 1 $ and $ \delta \geq 1 $.

We now argue that the approximation guarantee holds.
By Lemma~\ref{lem:incremental_algorithm_main_invariant}, we already know that
\begin{equation*}
\dist_G (x, \s) \leq \dist_{G_0} (x, \s) \leq \dist_G (x, \s) + A
%\dist_G (x, \s) \leq \dist_T (x, \s) \leq \dist_{G_0} (x, \s) \leq \dist_G (x, \s) + A
\end{equation*}
for every node $ x $ such that $ \dist_{G_0} (x, \s) \leq X $.
We now show that our choices of $ \kappa $ and~$ \delta $ guarantee that $ A \leq \epsilon \dist_G (x, \s) $, for every node $ x $ such that $ \dist_{G_0} (x, \s) \geq X/2 $, which immediately gives the desired inequality.

Assume that $ \dist_{G_0} (x, \s) \leq \dist_G (x, \s) + A $ and that $ \dist_{G_0} (x, \s) \geq X/2 $.
We first show that
\begin{equation*}
\gamma \leq \frac{1}{2 (1 + \frac{1}{\epsilon})} \, .
\end{equation*}
Since $ \epsilon \leq 1 $ we have $ 2 (\epsilon + 1) \leq 4 $.
It follows that
\begin{equation*}
\frac{1}{2 (1 + \frac{1}{\epsilon})} \geq \frac{\epsilon}{4} = \gamma \, .
\end{equation*}
Therefore we get the following chain of inequalities:
\begin{equation*}
\left( 1 + \frac{1}{\epsilon} \right) A = \left( 1 + \frac{1}{\epsilon} \right) \gamma X \leq \frac{\left( 1 + \frac{1}{\epsilon} \right) X}{2 (1 + \frac{1}{\epsilon})} = \frac{X}{2} \leq \dist_{G_0} (x, \s) \, .
\end{equation*}
We now subtract $ A $ from both sides and get
\begin{equation*}
\frac{A}{\epsilon} \leq \dist_{G_0} (x, \s) - A \, .
\end{equation*}
Since $ \dist_{G_0} (x, \s) - A \leq \dist_G (x, \s) $ by assumption, we finally get $ A \leq \epsilon \dist_G (x, \s) $.
\end{proof}

\subsection{Sequential model}\label{sec:incremental_sequential_model}

It is straightforward to use the abstract framework of Section~\ref{sec:incremental_abstract} in the sequential model.
First of all, note that in the sequential model computing a BFS tree takes time $ O (m) $, regardless of the depth.
We run $ O (\log{n}) $ ``parallel'' instances of Algorithm~\ref{alg:insertion_algorithm}, where each instance provides a $ (1 + \epsilon) $-approximation for nodes in some distance range from $ X/2 $ to $ X $.
However, when $ X $ is small enough, then instead of maintaining the approximate distance with our own algorithm it is more efficient to maintain the exact distance using the algorithm of Even and Shiloach~\cite{EvenS81}.

\begin{theorem}\label{thm:incremental_SSSP_RAM_model}
In the sequential model, there is an incremental $(1+\epsilon)$-approximate SSSP algorithm for inserting up to $ q $ edges that has a total update time of $ O (m n^{1/4} \sqrt{q} \log{n} / \sqrt{\epsilon}) $ where $ m $ is the number of edges in the final graph.
It answers distance and path queries in optimal worst-case time.
\end{theorem}

\begin{proof}
If $ q \leq 8 n^{1/2} / \epsilon $, we recompute a BFS tree from scratch after every insertion.
This takes time $ O (m q) = O (m q^{1/2} q^{1/2}) = O (m n^{1/4} q^{1/2} / \epsilon^{1/2}) $.

If $ q > 8 n^{1/2} / \epsilon $, the algorithm is as follows.
Let $ X^* $ be the smallest power of $ 2 $ greater than or equal to $ 2 n^{1/4} q^{1/2} / \epsilon^{1/2} $ (i.e., $ X^* = 2^{\lceil \log{(2 n^{1/4} q^{1/2} / \epsilon^{1/2})} \rceil} $).
First of all, we maintain an Even-Shiloach tree up to depth $ X^* $, which takes time $ O (m X^*) = O (m n^{1/4} q^{1/2} / \epsilon^{1/2}) $ by Theorem~\ref{thm:Even_Shiloach}.
Additionally, we run $ O (\log n) $ instances of Algorithm~\ref{alg:insertion_algorithm}, one for each $ \log{X^*} \leq i \leq \lceil \log{n} \rceil $.
For the $i$-th instance we set the parameter~$ X $ to $ X_i = 2^i $ and $ \kappa $~and~$ \delta $ as in Lemma~\ref{lem:multiplicative_approximation_abstract}.
Every time we start a new phase for instance $ i $, we also start a new phase for every instance $ j $ such that $ j \leq i $.
This guarantees that if a node leaves the range $ [X_i/2, X_i] $ (which in the incremental model can only happen if the distance to the root goes below $ X_i/2 $) it will immediately be covered by a lower range.
Since the graph is connected, we now have the following property:
for every node $ v $ with distance more than $ X^* $ to $ \s $ there is at least one index $ i $ such that $ v $ is in the range $ [X_i/2, X_i] $, i.e., at the beginning of the current phase of instance~$ i $ the distance from $ v $ to $ s $ was between $ X_i / 2 $ and $ X_i $.
By Lemma~\ref{lem:multiplicative_approximation_abstract} this previous distance is a $ (1+\epsilon) $-approximation of the current distance.
The algorithm can, at no overhead in asymptotic running time, easily track the smallest $ i $ such that $ v $ is in the range $ [X_i/2, X_i] $ for every node $ v $.

The cost of starting a new phase for every instance $ j \leq i $ is $ O (m \log{n}) $ since we have to construct a BFS tree up to depth $ X_j $ for all $ j \leq i $.
By Lemma~\ref{lem:multiplicative_approximation_abstract} the running time of the $i$-th instance of Algorithm~\ref{alg:insertion_algorithm} therefore is $ O (m q^{2/3} n^{1/3} \log{n} / (\epsilon^{2/3} X_i^{1/3})) $, which over all instances gives a running time of
\begin{equation*}
O \left( \sum_{\log X^* \leq i \leq \lceil \log n \rceil} \frac{m q^{2/3} n^{1/3} \log{n}}{\epsilon^{2/3} X_i^{1/3}} \right) = O \left( \frac{m q^{2/3} n^{1/3} \log{n}}{\epsilon^{2/3} {X^*}^{1/3}} \right) = O \left( \frac{m n^{1/4} q^{1/2} \log{n}}{\epsilon^{1/2}} \right) \, .
\end{equation*}
Note that for each instance $ i $ of Lemma~\ref{lem:multiplicative_approximation_abstract} only applies if $ \gamma^2 q X_i \geq n $ and $ \gamma n X_i^2 \geq q $.
These two inequalities hold because $ q $ and $ X^* $ are large enough:
\begin{gather*}
\gamma^2 q X_i = \epsilon^2 q X_i / 16 \geq \epsilon^2 q X^* / 16 \geq \epsilon^{3/2} q^{3/2} n^{1/4} / 8 \geq \epsilon^{3/2} n^{3/4} n^{1/4} / \epsilon^{3/2} = n \\
\gamma n X_i^2 = \epsilon n X_i^2 / 4 \geq \epsilon n (X^*)^2 / 4 \geq 4 \epsilon n^{3/2} q / (4 \epsilon) = n^{3/2} q \geq q
\end{gather*}

Finally, we argue that the number $ q $ of insertions does not have to be known beforehand.
We use a doubling approach for guessing the value of $ q $ where the $i$-th guess is $ q_i = 2^i $.
When the number of insertions exceeds our guess $ q_i $, we simply restart the algorithm and use the guess $ q_{i+1} = 2 q_i $ from now on.
The total running time for this approach is $ O (\sum_{i=0}^{\lceil \log q \rceil} m n^{1/4} q_i^{1/2} \log{n} / \epsilon^{1/2}) $ which is $ O (m n^{1/4} q^{1/2} \log{n} / \epsilon^{1/2}) $.
\end{proof}

\subsection{Distributed Model}\label{sec:incremental_distributed_model}

In the distributed model we use the same multi-layer approach as in the sequential model.
However, we have to consider some additional details for implementing the algorithm because not all information is globally available to every node in the distributed model.
Computing a BFS tree up to depth $ X $ takes time $ T_{\mathrm{BFS}} = O(X) $ in the distributed model.
In the running time analysis of Lemma~\ref{lem:running_time_incremental_abstract} we thus charge time $ O (X) $ to every phase and constant time to every insertion.
We now argue that this is enough to implement the algorithm in the distributed model.

After the insertion of an edge $ (u, v) $ the nodes $ u $ and $ v $ have to compare their initial distances $ \dist_{G_0} (u, \s) $ and $ \dist_{G_0} (v, \s) $.
They can exchange these numbers with a constant number of messages which we account for by charging constant time to every insertion.

The root node $ \s $ has to coordinate the phases and thus needs to gather some special information.
The first cause for starting a new phase is when the level of some node decreases by at least $ \delta $.
If a node detects a level decrease by at least $ \delta $, it has to inform the root $ \s $ about the decrease so that $ \s $ can initiate the beginning of the next phase.
The tree maintained by our algorithm, which has depth at most $ X $, can be used to send this message.
Therefore the total time needed for sending this message is $ O(X) $, which we charge to the current phase.
Note that, similar to recomputing the BFS tree, this happens in a recovery stage during which no new edges are inserted.

The second cause for starting a new phase is that the number of edge insertions since the beginning of the current phase exceeds $ \kappa $.
Therefore it is necessary that the root~$ \s $ knows the number of edges that have been inserted.
We count the number of insertions at the root as follows.
After each insertion of an edge $ (u, v) $ the node $ v $ sends a message to the root to inform it about the edge insertion.
We will make sure that this message arrives at the root with small enough delay; in particular each insertion message will arrive at the root after $ \kappa / 2 $ recovery stages.
Again, the tree maintained by our algorithm, which has depth at most $ X $, can be used to send the insertion messages to the root.
During each recovery stage we move up all the insertion messages that have not yet arrived at the root along $ 2 X / \kappa $ nodes in the tree (i.e., we decrease the level of each such message by at least $ 2 X / \kappa $).
To avoid congestion we aggregate insertion messages meeting at the same node by simply counting the \emph{number} of insertions.
Thus, we need to spend an additional $ O (X / \kappa) $ rounds in each recovery stage.
In this way, the first insertion message arrives at the root after $ \kappa/2 $ recovery stages and after $ \kappa $ recovery stages the first $ \kappa/2 $ messages have arrived
Accumulated over $ \kappa $ recovery stages after insertions, the total time for sending the insertion messages is $ O(\kappa X / \kappa) = O(X) $, which we charge to the current phase.
Thus, to get the same approximation guarantee and the same asymptotic running time as in Section~\ref{sec:incremental_distributed_model}, we slightly modify the algorithm to start a new phase as soon as the root has been notified of $ \kappa / 2 $ insertions.

\begin{theorem}\label{thm:incremental_SSSP_distributed_model}
In the distributed model, there is an incremental algorithm for maintaining a $ (1+\epsilon) $-approximate BFS tree under up to $ q $ insertions that has a total update time of
$ O ( q^{2/3} n^{1/3} D^{2/3} / \epsilon^{2/3} ) $,
where $ D $ is the dynamic diameter.
\end{theorem}

\begin{proof}
Our algorithm consists of $ O (\log{D}) $ layers.
For each $ 0 \leq i \leq \lceil \log{D} \rceil $ we set $ X_i = 2^i $ and do the following:
(1) If $ q \leq 16 n / (\epsilon^2 X_i) $, we recompute a BFS tree up to depth $ X_i $ from scratch after every insertion.
(2) If $ q > 16 n / (\epsilon^2 X_i) $ and $ X_i \leq 4 \sqrt{q} / \sqrt{\epsilon n} $, we maintain an Even-Shiloach tree up to depth $ X_i $.
(3) If $ q > 16 n / (\epsilon^2 X_i) $ and $ X_i > 4 \sqrt{q} / \sqrt{\epsilon n} $ we run an instance of Algorithm~\ref{alg:insertion_algorithm} with parameters $ X_i = 2^i $ and $ \kappa_i $ and $ \delta_i $ as in Lemma~\ref{lem:multiplicative_approximation_abstract}.
We use the following slight modification of Algorithm~\ref{alg:insertion_algorithm}:
every time a new phase starts for instance $ i $, we re-initialize all instances $ j $ of Algorithm~\ref{alg:insertion_algorithm} such that $ j \leq i $ by computing a BFS tree up to depth $ X_j $.
Note that if $ D $ is not known in advance, our algorithm can simply increase the number of layers until the BFS tree computed at the initialization of the current layer contains all nodes of the graph.

We first argue that this algorithm provides a $ (1+\epsilon) $-approximation.
The algorithm maintains the exact distances for all nodes that are in distance at most $ 16 n / (\epsilon^2 q) $ or $ 4 \sqrt{q} / \sqrt{\epsilon n} $ from the root as in these cases the distances are obtained by recomputing the BFS tree from scratch or by the Even-Shiloach tree.
For all other nodes we have to argue that our multi-layer version of Algorithm~\ref{alg:insertion_algorithm} provides a $ (1+\epsilon) $-approximation.
Note that for each instance $ i $ the result of Lemma~\ref{lem:multiplicative_approximation_abstract} only applies if $ \gamma^2 q X_i \geq n $ and $ \gamma n X_i^2 \geq q $.
These two inequalities hold because $ q $ and $ X_i $ are large enough:
\begin{gather*}
\gamma^2 q X_i = \epsilon^2 q X_i / 16 \geq \epsilon^2 (16 n / (\epsilon^2 X_i)) X_i / 16 = n \\
\gamma n X_i^2 = \epsilon n X_i^2 / 4 \geq  \epsilon n (4 \sqrt{q} / \sqrt{\epsilon n})^2 / 4 = q \, .
\end{gather*}
In each instance $ i $, the approximation guarantee of Lemma~\ref{lem:multiplicative_approximation_abstract} holds for all nodes whose distance to the root was between $ X_i / 2 $ and $ X_i $ since the last initialization of instance~$ i $.
Every time we re-initialize instance $ i $, some nodes that before were in the range $ [X_i/2, X_i] $ might now have a smaller distance and will thus not be ``covered'' by instance $ i $ anymore.
By re-initializing all instances $ j \leq i $ as well we guarantee that such nodes will immediately be ``covered'' by some other instance of the algorithm (or by the exact BFS tree we maintain for small depths).
Since the graph is connected, we thus have the following property:
for every node $ v $ with distance more than $ 4 \sqrt{q} / \sqrt{\epsilon n} $ to the root there is an index $ i $ such that at the beginning of the current phase of instance~$ i $ the distance from $ v $ to the root was between $ X_i / 2 $ and $ X_i $.
By Lemma~\ref{lem:multiplicative_approximation_abstract} this previous distance is a $ (1+\epsilon) $-approximation of the current distance.

We will now bound the running time.
We will argue that the running time in every layer $ i $ is $ O ( q^{2/3} n^{1/3} X_i^{2/3} / \epsilon^{2/3} ) $.
If the number of insertions is at most $ q \leq 16 n / (\epsilon^2 X_i) $, then computing a BFS tree from scratch up to depth $ X_i $ after very insertion takes time $ O (q X_i) $ in total, which we can bound as follows:
\begin{equation*}
q X_i = q^{2/3} q^{1/3} X_i = \frac{q^{2/3} n^{1/3} X_i^{2/3}}{\epsilon^{2/3}} \, .
\end{equation*}
By Theorem~\ref{thm:Even_Shiloach} maintaining an Even-Shiloach tree up to depth $ X_i \leq 4 \sqrt{q} / \sqrt{\epsilon n} $ takes time $ O (n X_i) = O (\sqrt{q n} / \sqrt{\epsilon}) $.
Since we only do this in the case $ q > 16 n / (\epsilon^2 X_i) $, we can use the inequality
\begin{equation*}
n < \frac{\epsilon^2 q X_i}{16} \leq \frac{q X_i^4}{\epsilon}
\end{equation*}
to obtain
\begin{equation*}
n X_i = \frac{\sqrt{q n}}{\sqrt{\epsilon}} = \frac{n^{1/3} n^{1/6} \sqrt{q}}{\sqrt{\epsilon}} \leq \frac{n^{1/3} q^{1/6} X_i^{4/6} \sqrt{q}}{\sqrt{\epsilon} \cdot \epsilon^{1/6}} = \frac{q^{2/3} n^{1/3} X_i^{2/3}}{\epsilon^{2/3}} \, .
\end{equation*}
Finally we bound the running time of our slight modification of Algorithm~\ref{alg:insertion_algorithm} in layer~$ i $.
Every time we start a new phase in layer $ i $, we re-initialize the instances of Algorithm~\ref{alg:insertion_algorithm} in all layers $ j \leq i $.
The re-initialization takes in each layer $ j $ takes time $ O (X_j) $ as we have to compute a BFS tree up to depth $ X_j $.
Thus, the cost of starting a new phase in layer $ i $ is proportional to
\begin{equation*}
\sum_{0 \leq j \leq i} X_j = \sum_{0 \leq j \leq i} 2^j \leq 2^{i+1} = 2 X_i
\end{equation*}
which asymptotically is the same as only the time needed for computing a BFS tree up to depth~$ X_i $.
Thus, by Lemma~\ref{lem:multiplicative_approximation_abstract} the running time of instance~$ i $ of Algorithm~\ref{alg:insertion_algorithm} is $ O (q^{2/3} n^{1/3} X_i^{2/3} / \epsilon^{2/3} + q) $.
Since $ q \leq \epsilon n X_i^2 / 4 $ as argued above we have $ q \leq n X_i^2 $ and thus
\begin{equation*}
\frac{q^{2/3} n^{1/3} X_i^{2/3}}{\epsilon^{2/3}} \geq q^{2/3} n^{1/3} X_i^{2/3} = q^{2/3} (n X_i^2)^{1/3} \geq q^{2/3} q^{1/3} = q \, .
\end{equation*}
It follows that the running time of instance~$ i $ is $ O (q^{2/3} n^{1/3} X_i^{2/3}/ \epsilon^{2/3}) $ and the total running time over all layers is
\begin{equation*}
O \left( \sum_{0 \leq i \leq \lceil \log{D} \rceil} \frac{q^{2/3} n^{1/3} X_i^{2/3}}{\epsilon^{2/3}} \right) = O \left( \sum_{0 \leq i \leq \lceil \log{D} \rceil} \frac{q^{2/3} n^{1/3} (2^i)^{2/3}}{\epsilon^{2/3}} \right) = O \left( \frac{q^{2/3} n^{1/3} D^{2/3}}{\epsilon^{2/3}} \right) \, .
\end{equation*}

By using a doubling approach for guessing the value of $ q $ we can run the algorithm with the same asymptotic running time without knowing the number of insertions beforehand.
\end{proof}

\subsection{Removing the Connectedness Assumption}\label{sec:removing_connectedness_assumption_incremental}

The algorithm above assumes that the graph is connected.
We now explain how to adapt the algorithm to handle graphs where this is not the case.

Note that an insertion might connect one or several nodes to the root node.
For each newly connected node, every path to the root goes through an edge that has just been inserted.
In such a situation we extend the tree maintained by the Algorithm~\ref{alg:insertion_algorithm} by performing a breadth-first search among the newly connected nodes.
Using this modified tree, the argument of Lemma~\ref{lem:incremental_algorithm_main_invariant} to prove the additive approximation guarantee still goes through.
Note that each node can become connected to the root node at most once.
Thus, we can amortize the cost of the breadth-first searches performed to extend the tree over all insertions.

This results in the following modification of the running time of Lemma~\ref{lem:running_time_incremental_abstract}: In the sequential model we have an additional cost of $ O (m) $ as each edge has to be explored at most once in one of the breadth-first searches.
In the distributed model we have an additional cost of $ O (n) $ as every node is explored at most once in one of the breadth-first searches.
The total update time of the $ (1 + \epsilon) $-approximation in the sequential model (Theorem~\ref{thm:incremental_SSSP_RAM_model}) clearly stays unaffected from this modification as we anyway have to consider the cost of $ O (m) $ for computing a BFS tree.
In the distributed model the argument is as follows.
In the proof of Theorem~\ref{thm:incremental_SSSP_distributed_model} we bound the running time of each instance $ i $ of Algorithm~\ref{alg:insertion_algorithm} by $ O (q^{2/3} n^{1/3} X_i^{2/3} / \epsilon^{2/3}) $.
Since $ q $ and~$ X_i $ satisfy the inequality $ q > 16 n / (\epsilon^2 X_i) \geq n / X_i $, we have $ q^{2/3} n^{1/3} X_i^{2/3} / \epsilon^{2/3} \geq q^{2/3} n^{1/3} X_i^{2/3} \geq n $.
Thus the additional $ O (n) $ is already dominated by $ O (q^{2/3} n^{1/3} X_i^{2/3} / \epsilon^{2/3}) $ and the total update time stays the same as stated in Theorem~\ref{thm:incremental_SSSP_distributed_model}.
Note that if the number of nodes $ n $ is not known in advance because of the graph not being connected, we can use a doubling approach to guess the right range of $ n $.

\section{Decremental Algorithm}\label{sec:decremental_algorithm}

In the decremental setting we use an algorithm of the same flavor as in the incremental setting (see Algorithm~\ref{alg:decremental_algorithm}).
However, the update procedure is more complicated because it is not obvious how to repair the tree after a deletion.
Our solution exploits the fact that in the distributed model it is relatively cheap to examine the local neighborhood of a node.
As in the incremental setting, the algorithm has the parameters $ \kappa $, $ \delta $, and $ X $.

\begin{algorithm}
\caption{Decremental algorithm}
\label{alg:decremental_algorithm}

\SetKwProg{procedure}{Procedure}{}{}
\SetKwFunction{initialize}{Initialize}
\SetKwFunction{delete}{Delete}
\SetKwFunction{repairTree}{RepairTree}

\tcp{At any time, $ T_0 $ is the BFS tree computed at the beginning of the current phase, $ F' $ is the forest resulting from removing all deleted edges from~$ T_0 $, and $ T $ is the current approximate BFS tree}

\BlankLine

\procedure{\delete{$u$, $v$}}{
	$ k \gets k+1 $\;
	\eIf{$ k = \kappa $}{
		\initialize{}\;
	}{
		Remove edge $ (u, v) $ from $ F' $\;
		\repairTree{}\;
		\lIf{\repairTree{} reports distance increase by at least $ \delta $}{\initialize{}}
	}
}

\BlankLine

\procedure(\tcp*[f]{Start new phase}){\initialize{}}{
	$ k \gets 0 $\;
	Compute BFS tree $ T_0 $ of depth $ X $ rooted at $ \s $\; % Set $ F \gets T $\;
	Compute current distances $ \dist_{G_0} (\cdot, \s) $\;
	$ T \gets T_0 $\;
	$ F' \gets T_0 $\;
}

\BlankLine

\procedure{\repairTree{}}{
	$ F \gets F' $\;
	$ U \gets \{ u \in V \mid \text{$ u $ has no parent in $ F $ and $ u \neq s $} \} $\;
	\ForEach(\tcp*[f]{Search process}){$ u \in U $}{
		Perform breadth-first search from $ u $ up to depth $ \delta $ and try to find a node $ v $ such that (1) $ \dist_{G_0} (v, \s) < \dist_{G_0} (u, \s) $ and (2) $ \dist_{G} (u, v) \leq \delta $\;\label{lin:conditions_good_node}
		\eIf{such a node $ v $ could be found}{
			Add edge $ (u, v) $ of weight $ \dist_{F}^\mathrm{w} (u, v) = \dist_{G} (u, v) $ to $ F $\;
		}{
			\Return{``distance increase by at least $ \delta $''}\;
		}
	}
	$ T \gets F $\;
}
\end{algorithm}

The Procedure \repairTree of Algorithm~\ref{alg:decremental_algorithm} either computes a (weighted) tree $ T $ that approximates the true distances with an additive error of $ \kappa \delta $, or it reports a distance increase by at least $ \delta $ since the beginning of the current phase.
Let $ T_0 $ denote the BFS tree computed at the beginning of the current phase, let $ F' $ be the forest resulting from removing those edges from $ T_0 $ that have already been deleted in the current phase, and the let $ U $ be the set of nodes (except for~$ \s $) that have no parent in $ F' $.
After every deletion, the Procedure \repairTree tries to construct a tree $ T $ by starting with the forest~$ F' $.
Every node $ u \in U $ tries to find a ``good'' node~$ v $ to reconnect to and if successful will use $ v $ as its new parent with a weighted edge $ (u, v) $ (whose weight corresponds to the current distance between $ u $ and $ v $).
Algorithm~\ref{alg:decremental_algorithm} imposes two conditions (Line~\ref{lin:conditions_good_node}) on a ``good'' node $ v $.
Condition~(1) avoids that the reconnection introduces any cycles and Condition~(2) guarantees that the error introduced by each reconnection is at most~$ \delta $ and that a suitable node $ v $ can be found in distance at most $ \delta $ to $ u $.
As $ \delta $ is relatively small, this is the key to efficiently finding such a node.
In the following, we denote the distance between two nodes $ x $ and $ y $ in a graph $ F $ with weighted edges by $ \dist_{F}^\mathrm{w} (x, y) $.
Note that here we have formulated the algorithm in a way such that the Procedure \repairTree always starts with a forest $ F' $ that is the result of removing all edges from~$ T_0 $ that have been deleted so far in the current phase, regardless of trees previously computed by the Procedure \repairTree.

\subsection{Analysis of Procedure for Repairing the Tree}

In the following we first analyze only the Procedure \repairTree.
Its guarantees can be summarized as follows.
\begin{lemma}\label{lem:approximation}
The Procedure \repairTree of Algorithm~\ref{alg:decremental_algorithm} either reports ``distance increase by at least $ \delta $'' and guarantees that there is a node $ x $ with $ \dist_{G_0} (x, \s) \leq X $ such that
\begin{equation*}
\dist_{G} (x, \s) \geq \dist_{G_0} (x, \s) + \delta \, ,
\end{equation*}
or it returns a tree $ T $ such that for every node $ x $ with $ \dist_{G_0} (x, \s) \leq X $ we have
\begin{equation*}
\dist_{G_0} (x, \s) \leq \dist_{G} (x, \s) \leq \dist_{T}^\mathrm{w} (x, \s) \leq \dist_{G_0} (x, \s) + \kappa \delta \, .
\end{equation*}
It runs in time $ O (\kappa \delta) $ after every deletion.
\end{lemma}

We first observe that if the Procedure \repairTree returns a graph, this graph is actually a tree.
The input of the procedure is the forest $ F' $ obtained from removing some edges from the BFS tree $ T_0 $.
In this forest we have $ \dist_{G_0} (v, \s) = \dist_{G_0} (u, \s) - 1 $ for every child $ u $ and parent $ v $.
In the Procedure \repairTree, we add, for every node $ u $ that is missing a parent, an edge to a parent $ v $ such that $ \dist_{G_0} (v, \s) < \dist_{G_0} (u, \s) $.
Thus, the decreasing label $ \dist_{G_0} (v, \s) $ for every node $ v $ guarantees that $ T $ is a tree.
\begin{lemma}\label{pro:efficient_algorithm_forest}
The graph $ T $ computed by the Procedure \repairTree is a tree.
\end{lemma}

We will show next that the Procedure $ \repairTree $ is either successful, i.e., every node in $ U $ finds a new parent, or the algorithm makes progress because there is some node whose distance to the root has increased significantly.
\begin{lemma}\label{lem:efficient_algorithm_existence_of_good_node}
For every node $ u \in U $, if $ \dist_{G} (u, \s) < \dist_{G_0} (u, \s) + \delta $, then there is a node $ v \in V $ such that
\begin{enumerate}[label=(\arabic*)]
\item $ \dist_{G_0} (v, \s) < \dist_{G_0} (u, \s) $
and
\item $ \dist_{G} (u, v) \leq \delta $.
\end{enumerate}
\end{lemma}

\begin{proof}
If $ \dist_G (u, \s) \leq \delta $, then set $ v = \s $.
As $ \dist_{G_0} (\s, \s) = 0 $ and $ u \neq \s $, this satisfies both conditions.

If $ \dist_G (u, \s) > \delta $, then consider the shortest path from $ u $ to $ \s $ in $ G $ and define $ v $ as the node that is in distance $ \delta $ from $ u $ on this path, i.e., such that $ \dist_G (v, \s) = \dist_G (u, \s) - \delta $.
We then have
\begin{equation*}
\dist_{G_0} (v, \s) \leq \dist_G (v, \s) = \dist_G (u, \s) - \delta < \dist_{G_0} (u, \s) + \delta - \delta = \dist_{G_0} (u, \s) \, . \qedhere
\end{equation*}
\end{proof}

Note that in the proof above we know exactly which node $ v $ we can pick for every node $ u \in U $.
In the algorithm however the node $ u $ does not know its shortest path to $ \s $ in the current graph and thus it would be expensive to specifically search for the node~$ v $ on the shortest path from $ u $ to $ \s $ defined above.
However, we know that $ v $ is contained in the local search performed by $ u $.
Therefore $ u $ either finds $ v $ or some other node that fulfills Conditions~(1) and~(2).

We now show that every reconnection made by the Procedure \repairTree adds an additive error of $ \delta $, which sums up to $ \kappa \delta $ for at most $ \kappa $ reconnections (one per previous edge deletion).
\begin{lemma}\label{lem:approximation_forest}
For the tree $ T $ computed by the Procedure \repairTree and every node~$ x $ such that $ \dist_{G_0} (x, \s) \leq X $ we have
\begin{equation*}
\dist_{G} (x, \s) \leq \dist_T^\mathrm{w} (x, \s) \leq \dist_{G_0} (x, \s) + \kappa \delta \, .
\end{equation*}
\end{lemma}

\begin{proof}
We call the weighted edges inserted by the Procedure \repairTree \emph{artificial edges}.
In the tree $ T $ there are two types of edges: those that were already present in the BFS tree $ {T_0} $ from the beginning of the current phase and artificial edges added in the Procedure \repairTree.

First, we prove the inequality $ \dist_{G} (x, \s) \leq \dist_{T}^\mathrm{w} (x, \s) $.
Consider the unique path from $ x $ to $ \s $ in the tree $ T $ consisting of the nodes $ x = x_l, x_{l-1}, \ldots x_0 = s $.
We know that every edge $ (x_{j+1}, x_j) $ in $ T $ either was part of the initial BFS tree $ T_0 $, which means that $ \dist_{T}^\mathrm{w} (x_{j+1}, x_j) = 1 = \dist_G (x_{j+1}, x_j) $, or was inserted later by the algorithm, which means that $ \dist_{T}^\mathrm{w} (x_{j+1}, x_j) = \dist_G (x_{j+1}, x_j) $.
This means that in any case we have $ \dist_{T}^\mathrm{w} (x_{j+1}, x_j) = \dist_G (x_{j+1}, x_j) $ and therefore we get
\begin{equation*}
\dist_{T}^\mathrm{w} (x, \s) = \sum_{j=0}^{l-1} \dist_{T}^\mathrm{w} (x_{j+1}, x_j)
= \sum_{j=0}^{l-1} \dist_{G} (x_{j+1}, x_j) \geq \dist_{G} (x, \s) \, .
\end{equation*}

Second, we prove the inequality $ \dist_T^\mathrm{w} (x, \s) \leq \dist_{G_0} (x, \s) + \kappa \delta $.
Consider the shortest path $ \pi = x_l, x_{l-1}, \ldots, x_0 $ from $ x $ to $ \s $ in $ T $, where $ x_l = x $ and $ x_0 = \s $.
Let $ S_j $ (with $ 0 \leq j \leq l $) denote the number of artificial edges on the subpath $ x_j, x_{j-1}, \ldots x_0 $.
For each edge deleted so far, $ \pi $ contains at most one artificial edge.
Therefore we have $ S_j \leq \kappa $ for all $ 0 \leq j \leq l $.
Now consider the following claim.
\begin{claim}
For every $ 0 \leq j \leq l $ we have $ \dist_T^\mathrm{w} (x_j, \s) \leq \dist_{G_0} (x_j, \s) + S_j \delta $.
\end{claim}
Assuming the truth of the claim, the desired inequality follows straightforwardly since $ x_l = x $, $ x_0 = \s $, and $ S_l \leq \kappa $.

In the following we prove the claim by induction on $ j $.
In the induction base we have $ j = 0 $ and thus $ x_j = \s $ and $ S_j = 0 $.
The inequality then trivially holds due to $ \dist_T^\mathrm{w} (\s, \s) = 0 $.
We now prove the inductive step from $ j $ to $ j + 1 $.
Note that $ S_j \leq S_{j+1} \leq S_{j} + 1 $ since the path is exactly one edge longer.
Consider first the case that $ (x_{j+1}, x_j) $ is an edge from the BFS tree $ {T_0} $ of the graph $ G_0 $.
In that case we have $ \dist_T^\mathrm{w} (x_{j+1}, x_j) = \dist_{G_0} (x_{j+1}, x_j) = 1 $.
Furthermore, since $ (x_{j+1}, x_j) $ is an edge in the BFS tree $ {T_0} $ we know that $ x_j $ lies on a shortest path from $ x_{j+1} $ to $ \s $ in $ G_0 $.
Therefore we have $ \dist_{G_0} (x_{j+1}, \s) = \dist_{G_0} (x_{j+1}, x_j) + \dist_{G_0} (x_j, \s) $.
Together with the induction hypothesis we get:
\begin{align*}
\dist_T^\mathrm{w} (x_{j+1}, \s)
&= \underbrace{\dist_T^\mathrm{w} (x_{j+1}, x_j)}_{= \dist_{G_0} (x_{j+1}, x_j)} + \underbrace{\dist_T^\mathrm{w} (x_j, \s)}_{\leq \dist_{G_0} (x_j, \s) + S_j \cdot \delta \text{ (by IH)}} \\
&\leq \underbrace{\dist_{G_0} (x_{j+1}, x_j) + \dist_{G_0} (x_j, \s)}_{= \dist_{G_0} (x_{j+1}, \s)} + \underbrace{S_j}_{= S_{j+1}} \cdot \delta \\
&= \dist_{G_0} (x_{j+1}, \s) + S_{j+1} \cdot \delta \, .
\end{align*}

The second case is that $ (x_{j+1}, x_j) $ is an artificial edge.
In that case we have $ \dist_T^\mathrm{w} (x_{j+1}, x_j) = \dist_{G} (x_{j+1}, x_j) $ and by the algorithm the inequality $ \dist_{G} (x_{j+1}, x_j) + \dist_{G_0} (x_j, \s) \leq \dist_{G_0} (x_{j+1}, \s) + \delta $ holds.
Note also that $ S_{j+1} = S_j + 1 $.
We therefore get the following:
\begin{align*}
\dist_T^\mathrm{w} (x_{j+1}, \s)
&= \underbrace{\dist_T^\mathrm{w} (x_{j+1}, x_j)}_{= \dist_{G} (x_{j+1}, x_j)} + \underbrace{\dist_T^\mathrm{w} (x_j, \s)}_{\leq \dist_{G_0} (x_j, \s) + S_j \cdot \delta \text{ (by IH)}} \\
&\leq \underbrace{\dist_{G} (x_{j+1}, x_j) + \dist_{G_0} (x_j, \s)}_{\leq \dist_{G_0} (x_{j+1}, \s) + \delta} + S_j \cdot \delta \\
&= \dist_{G_0} (x_{j+1}, \s) + \underbrace{(S_j + 1)}_{= S_{j+1}} \cdot \delta \\
&= \dist_{G_0} (x_{j+1}, \s) + S_{j+1} \cdot \delta \, . \qedhere
\end{align*}
\end{proof}

\begin{remark}
In the proof of Lemma~\ref{lem:approximation_forest} we need the property that after up to $ \kappa $ edge deletions there are at most $ \kappa $ ``artificial'' edges on the shortest path to the root in~$ T $.
This also holds if we allow deleting nodes (together with their set of incident edges).
Thus, we can easily modify our algorithm to deal with node deletions with the same approximation guarantee and asymptotic running time.
\end{remark}

To finish the proof of Lemma~\ref{lem:approximation} we analyze the running time of the Procedure \repairTree and clarify some implementation details for the distributed setting.
In the search process, every node $ u \in U $ tries to find a node $ v $ to connect to that fulfills certain properties.
We search for such a node $ v $ by examining the neighborhood of $ u $ in~$ G $ up to depth $ \delta $ using breadth-first search, which takes time $ O (\delta) $ for a single node.
Whenever local searches of nodes in $ U $ ``overlap'' and two messages have to be sent over an edge, we arbitrarily allow to send one of these messages and delay the other one to the next round.
As there are at most $ \kappa $ nodes in $ U $, we can simply bound the time needed for all searches by $ O (\kappa \delta) $.

\paragraph{Weighted Edges.}
The tree computed by the algorithm contains weighted edges. %, which do not exist in the network.
Such an edge $ e $ corresponds to a path $ \pi $ of the same distance in the network.
We implement weighted edges by a routing table for every node $ v $ that stores the next node on $ \pi $ if a message is sent over $ v $ as part of the weighted edge $ e $.

\paragraph{Broadcasting Deletions.}
The nodes that do not have a parent in $ F' $ before the procedure \repairTree starts do not necessarily know that a new edge deletion has happened.
Such a node only has to become active and do the search if there is a change in its neighborhood within distance $ \delta $, otherwise it can still use the weighted edge in the tree $ T $ that it previously used because the two conditions imposed by the algorithm will still be fulfilled.
After the deletion of an edge $ (x, y) $, the nodes $ x $ and $ y $ can inform all nodes at distance $ \delta $ about this event.
This takes time $ O (\delta) $ per deletion, which is within our projected running time.

\subsection{Analysis of Decremental Distributed Algorithm}

The Procedure \repairTree provides an additive approximation of the shortest paths and a means for detecting that the distance of some node to $ \s $ has increased by at least~$ \delta $ since the beginning of the current phase.
Using this procedure as a subroutine we can provide a running time analysis for the decremental algorithm that is very similar to the one of the incremental algorithm.

\begin{lemma}\label{lem:running_time_decremental_distributed}
For every $ X \geq 1 $, $ \kappa \geq 1 $, and $ \delta \geq 1 $ the total update time of Algorithm~\ref{alg:decremental_algorithm} is
$ O ( \q X / \kappa + n X^2 / \delta^2 + \q \kappa \delta + n) $ and it provides the following approximation guarantee: If $ \dist_{G_0} (x, \s) \leq X $, then
\begin{equation*}
\dist_{G_0} (x, \s) \leq \dist_{G} (x, \s) \leq \dist_{T}^\mathrm{w} (x, \s) \leq \dist_{G_0} (x, \s) + \kappa \delta \, .
\end{equation*}
\end{lemma}

\begin{proof}
Using the distance increase argument of Lemma~\ref{lem:distance_increase}, we can bound the number of phases by $ O (\q / \kappa + n X^2 / \delta^2) $.
To every phase, we charge a running time of $ O (X) $, which is the time needed for computing a BFS tree up to depth $ X $ at the beginning of the phase.
Additionally we charge a running time of $ \kappa \delta $ to every deletion since the Procedure \repairTree, which is called after every deletion, has a running time of $ O (\kappa \delta) $ by Lemma~\ref{lem:approximation}.

As in the incremental distributed algorithm we have to enrich the decremental algorithm with a mechanism that allows the root node to coordinate the phases.
We explain these implementation details and analyze their effects on the running time in the following.

\paragraph{Reporting Distance Increase.}
When a node $ v $ detects a distance increase by more than $ \delta $, it tries to inform the root about the distance increase by sending a special message.
It sends the message to all nodes in distance at most $ 2 X $ from $ v $ in a breadth-first manner, which takes time $ O (X) $.
If the root is among these nodes, the root initiates a new phase and the cost of $ O (X) $ is charged to the new phase.
Otherwise, the nodes in distance at most $ X $ from $ v $ know that their distance to the root is more than $ X $.
In that case in particular all nodes in the subtree of $ v $ in $ F' $ have received the message and know that their distance to the root is more than $ X $ now.
All nodes that are in distance at most $ X $ from $ v $ do not have to participate in the algorithm anymore.
Thus, we can charge the time of $ O (X) $ for sending the message to these at least $ X $ nodes.
This give a one-time charge of $ O(1) $ to every node and adds $ O(n) $ to the total update time.
A special case is when $ v $ becomes disconnected from the root and its new component has size less than $ X $.
In that case the time for sending the message to all nodes in the component takes time proportional to the size of the component, which again results in a charge of $ O (1) $ to each node.

\paragraph{Counting Deletions.}
We count the number of deletions at the root as follows.
First, observe that we do not have to count those deletions that result in a distance increase by more than~$ \delta $ because after such an event either a new phase starts or the deletion only affects nodes whose distance to the root has increased to more than~$ X $ after the deletion.
The remaining deletions can be counted by sending one message per deletion to the root using the tree maintained by the algorithm such that each deletion message will arrive at the root after $ \kappa / 2 $ recovery stages.
During each recovery stage we move up all the deletion messages that have not yet arrived at the root along $ 2 X / \kappa $ nodes in the tree.
To avoid congestion we aggregate deletion messages meeting at the same node by simply counting the \emph{number} of deletions.
Note that the level of a node in the tree might increase by at most $ \kappa \delta $ with every deletion.
Therefore we need spend time $ O (X/\kappa + \kappa \delta) $ during each recovery stage to ensure that every deletion message that has not yet arrived at the root decreases its level in the tree by at least $ 2 X/\kappa $.
In this way, the first deletion message arrives after $ \kappa/2 $ reocvery stages and after $ \kappa $ recovery stages the first $ \kappa/2 $ messages have arrived at the root.
This process takes total time $ O(X + \kappa^2 \delta) $ for $ \kappa $ recovery stages after deletions.
We can charge time $ O(X) $ to the current phase and time $ O (\kappa \delta) $ to each deletion ocurring in the phase.
Thus, to obtain an additive approximation of exactly $ \kappa \delta $, we slightly modify the algorithm to start a new phase as soon as the root has been notified of $ \kappa / 2 $ deletions.
\end{proof}

We use a similar approach as in the incremental setting to get the $ (1+\epsilon) $-approximation.
We run $ i $ ``parallel'' instances of the algorithm where each instance covers the distance range from $ 2^i $ to $ 2^{i+1} $.
By an appropriate choice of the parameters $ \kappa $ and $ \delta $ for each instance we can guarantee a $ (1+\epsilon) $-approximation.

\begin{lemma}\label{lem:multiplicative_approximation_decremental_distributed}
Let $ 0 < \epsilon \leq 1 $ and assume that $ \epsilon^5 q X / 32 \geq n $ and $ n X^{3/2} \geq \q $.
Then, by setting $ \kappa = \q^{1/5} X^{1/5} / n^{1/5} $ and $ \delta = n^{2/5} X^{3/5} / \q^{2/5} $, Algorithm~\ref{alg:decremental_algorithm} runs in time $ O ( \q^{4/5} n^{1/5} X^{4/5} ) $.
Furthermore, it provides the following approximation guarantee:
For every node $ x $ such that $ \dist_{G_0} (x, \s) \leq X $ we have 
\begin{equation*}
\dist_{G_0} (x, \s) \leq \dist_{G} (x, \s) \leq \dist_{T}^\mathrm{w} (x, \s)
\end{equation*}
and for every node $ x $ such that $ \dist_G (x, \s) \geq X/2 $ we additionally have
\begin{equation*}
\dist_{T}^\mathrm{w} (x, \s) \leq (1 + \epsilon) \dist_{G_0} (x, \s) \leq (1 + \epsilon) \dist_G (x, \s) \, .
\end{equation*}
\end{lemma}

\begin{proof}
Since $ \epsilon^5 q X / 32 \geq n $ implies $ q X \geq n $ we have $ \kappa \geq 1 $ and since $ n X^{3/2} \geq q $ we have $ \delta \geq 1 $.
It is easy to check that by our choices of $ \kappa $ and $ \delta $ the three terms in the running time of Lemma~\ref{lem:running_time_decremental_distributed} are balanced and we get:
\begin{equation*}
\frac{q}{\kappa} \cdot X = \frac{n X}{\delta^2} \cdot X = \q \kappa \delta = \q^{4/5} n^{1/5} X^{4/5} \, .
\end{equation*}
Furthermore, since $ \q X \geq n $ we have $ \q^{4/5} n^{1/5} X^{4/5} \geq n^{1/5} (\q X)^{4/5} \geq n^{1/5} n^{4/5} = n $ and therefore the running time of the algorithm is $ O (\q^{4/5} n^{1/5} X^{4/5}) $.

We now argue that the approximation guarantee holds.
By Lemma~\ref{lem:running_time_decremental_distributed}, we already know that
\begin{equation*}
\dist_{G_0} (x, \s) \leq \dist_{G} (x, \s) \leq \dist_{T}^\mathrm{w} (x, \s) \leq \dist_{G_0} (x, \s) + \kappa \delta
\end{equation*}
for every node $ x $ such that $ \dist_{G_0} (x, \s) \leq X $.
We now show that our choices of $ \kappa $ and~$ \delta $ guarantee that $ \kappa \delta \leq \epsilon \dist_{G_0} (x, \s) $, for every node $ x $ such that $ \dist_{G_0} (x, \s) \geq X/2 $, which immediately gives the desired inequality.
By our assumptions we have $ n \leq \epsilon^5 \q X /32 $ and therefore we get
\begin{equation*}
\kappa \delta = \frac{\q^{1/5} X^{1/5}}{n^{1/5}} \cdot \frac{n^{2/5} X^{3/5}}{\q^{2/5}} = \frac{n^{1/5} X^{4/5}}{\q^{1/5}} \leq \frac{\epsilon \q^{1/5} X^{1/5} X^{4/5}}{2 \q^{1/5}} = \frac{\epsilon X}{2} \leq \epsilon \dist_{G_0} (x, \s) \, .
\end{equation*}
The value $ (1 + \epsilon) \dist_{G_0} (x, \s) $ is thus a $ (1 + \epsilon) $-approximation of $ \dist_G (x, \s) $.
\end{proof}

\begin{theorem}\label{thm:decremental_multiplicative_approximation}
In the distributed model, there is a decremental algorithm for maintaining a $ (1+\epsilon) $-approximate BFS tree over $ \q $ deletions with a total update time of $ O ( q^{4/5} n^{1/5} D^{4/5} / \epsilon ) $, where $ D $ is the dynamic diameter.
\end{theorem}

\begin{proof}
Our algorithm consists of $ O (\log{D}) $ layers.
For each $ 0 \leq i \leq \lceil \log{D} \rceil $ we set $ X_i = 2^i $ and do the following:
If $ q \leq 32 n / (\epsilon^5 X_i) $, we recompute a BFS tree up to depth $ X_i $ from scratch after every deletion.
If $ q > 32 n / (\epsilon^5 X_i) $ and $ X_i \leq (q / n)^{2/3} $, we maintain an Even-Shiloach tree up to depth $ X_i $.
If $ q > 32 n / (\epsilon^5 X_i) $ and $ X_i > (q / n)^{2/3} $ we run an instance of Algorithm~\ref{alg:decremental_algorithm} with parameters $ X_i = 2^i $ and $ \kappa_i $ and $ \delta_i $ as in Lemma~\ref{lem:multiplicative_approximation_decremental_distributed}.
Note that $ D $ might increase over the course of the algorithm due to edge deletions (or might not be known in advance).
Therefore, whenever we initialize the algorithm in the layer with the current largest index, we do a full BFS tree computation.
If the depth of the BFS tree exceeds $ X_i $, we increase the number of layers accordingly and charge the running time of the BFS tree computation to the layer with new largest index.

We first argue that this algorithm provides a $ (1+\epsilon) $-approximation.
The algorithm maintains the exact distances for all nodes that are in distance at most $ 32 n / (\epsilon^5 q) $ or $ (q / n)^{2/3} $ from the root as in these cases the distances are obtained by recomputing the BFS tree from scratch or by the Even-Shiloach tree.
For all other nodes we have to argue that our multi-layer version of Algorithm~\ref{alg:decremental_algorithm} provides a $ (1+\epsilon) $-approximation.
Note that the approximation guarantee of Lemma~\ref{lem:multiplicative_approximation_abstract} only applies if $ \epsilon^5 q X_i / 32 \geq n $ and $ n X_i^{3/2} \geq q $.
These two inequalities hold because $ q $ and $ X_i $ are large enough:
\begin{gather*}
\epsilon^5 q X_i / 32 \geq \epsilon^5 (32 n / (\epsilon^5 X_i)) X_i / 32 = n \\
n X_i^{3/2} \geq n ((q / n)^{2/3})^{3/2} = q \, .
\end{gather*}
In each instance $ i $ of Algorithm~\ref{alg:decremental_algorithm}, the approximation guarantee of Lemma~\ref{lem:multiplicative_approximation_abstract} holds for all nodes whose distance to the root at the beginning of the current phase of instance~$ i $ was at most $ X_i $ and whose current distance to the root is at least $ X_i / 2 $.
Whenever an instance $ i $ starts a new phase, there might be some nodes who before were contained in the tree of instance $ i $, but are not contained in the new tree anymore because their distance to the root has increased to more than~$ X_i $.
Since $ X_i = X_{i+1} / 2 $ we know that those node will immediately be ``covered'' by an instance with larger index.
Thus, after each recovery stage every node that is connected to the root will be contained in the tree of some instance $ i $ such that the preconditions of Lemma~\ref{lem:multiplicative_approximation_abstract} apply and thus the distance to the root in that tree provides a $ (1+\epsilon) $-approximation.
In particular each node simply has to pick the tree of the smallest index containing it.

We will now bound the running time.
We will argue that the running time in every layer $ i $ is $ O ( \q^{4/5} n^{1/5} X_i^{4/5} / \epsilon ) $.
If the number of insertions is at most $ q \leq 32 n / (\epsilon^5 X_i) $, then computing a BFS tree from scratch up to depth $ X_i $ after very insertion takes time $ O (q X_i) $ in total, which we can bound as follows:
\begin{equation*}
\q X_i = \q^{4/5} q^{1/5} X_i \leq \frac{\q^{4/5} 32^{1/5} n^{1/5} X_i^{4/5}}{\epsilon} = O \left( \frac{\q^{4/5} n^{1/5} X_i^{4/5}}{\epsilon} \right) \, .
\end{equation*}
By Theorem~\ref{thm:Even_Shiloach} maintaining an Even-Shiloach tree up to depth $ X_i \leq (q / n)^{2/3} $ takes time $ O (n X_i) = O (q^{2/3} n^{1/3}) $.
Since we only do this in the case $ q > 32 n / (\epsilon^5 X_i) $, we can use the inequality
\begin{equation*}
n < \frac{\epsilon^5 q X_i}{32} \leq \epsilon^5 q X_i \leq \frac{q X_i^6}{\epsilon^{15/2}}
\end{equation*}
to obtain
\begin{equation*}
n X_i \leq q^{2/3} n^{1/3} = q^{2/3} n^{1/5} n^{2/15} \leq q^{2/3} n^{1/5} \frac{q^{2/15} X_i^{4/5}}{\epsilon} = \frac{q^{4/5} n^{1/5} X_i^{4/5}}{\epsilon} \, .
\end{equation*}
Finally we use Lemma~\ref{lem:multiplicative_approximation_decremental_distributed} to bound the running time of Algorithm~\ref{alg:decremental_algorithm} in layer~$ i $ by $ O ( \q^{4/5} n^{1/5} X_i^{4/5} / \epsilon ) $ as well.
Thus, the running time over all layers is
\begin{equation*}
O \left( \sum_{0 \leq i \leq \lceil \log{D} \rceil} \frac{\q^{4/5} n^{1/5} X_i^{4/5}}{\epsilon} \right)
 = O \left( \sum_{0 \leq i \leq \lceil \log{D} \rceil} \frac{\q^{4/5} n^{1/5} (2^i)^{4/5}}{\epsilon} \right)
 = O \left( \frac{\q^{4/5} n^{1/5} D^{4/5}}{\epsilon} \right) \, .
\end{equation*}
By using a doubling approach for guessing the value of $ q $ we can run the algorithm with the same asymptotic running time without knowing the number of deletions beforehand.
\end{proof}

\section{Conclusion and Open Problems}

In this paper, we showed that an approximate breadth-first search spanning tree can be maintained in amortized time per update that is sublinear in the diameter $D$ in partially dynamic distributed networks when amortized over a sufficient number of updates.
Many problems remain open. For example, can we get a similar result for the case of {\em fully-dynamic} networks? How about {\em weighted} networks (even partially dynamic ones)? Can we also get a sublinear time bound for the {\em all-pairs shortest paths} problem? Moreover, in addition to the sublinear-time complexity achieved in this paper, it is also interesting to obtain algorithms with small bounds on message complexity and memory.

We believe that the most interesting open problem is whether the sequential algorithm in this paper can be improved to obtain a deterministic incremental algorithm with near-linear total update time. As noted earlier, techniques from this paper have led to a {\em randomized} decremental algorithm with near-linear total update time \cite{HenzingerKN-FOCS14} (the same algorithm also works in the incremental setting). Whether this algorithm can be derandomized was left as a major open problem in~\cite{HenzingerKN-FOCS14}. As the incremental case is usually easier than the decremental case, it is worth obtaining this result in the incremental setting first.

\section*{Acknowledgement}
We thank the reviewers of ICALP 2013 for pointing to related papers and to an error in an example given in a previous version of this paper.
We also thank one of the reviewers of Transactions on Algorithms for very detailed comments.

\printbibliography[heading=bibintoc] % Make bibliography show up in table of contents

\end{document}